%% file: crbmain.tex
\documentclass[a4paper,notitlepage,10pt]{article}

\input{preamble}

\usepackage[numbers,sort]{natbib}

\newcommand{\add}[1]{{#1}}
\newcommand{\remove}[1]{}
\newcommand{\removesafe}[1]{}
\newcommand{\change}[2]{{#2}}
\newcommand{\changesafe}[2]{{#2}}

\usepackage{subfigure}

\begin{document}

\title{\LARGE Cramér-Rao bounds for synchronization of rotations}

\author{
{\sc By Nicolas Boumal}$^{\,1*}$, {\sc Amit Singer}$^{\,2}$,\\{\sc P.-A. Absil}$^{\,1}$ {\sc and Vincent D.\ Blondel}$^{\,1}$ \\[5pt]
\small (1) Department of Mathematical Engineering, ICTEAM Institute,\\ \small  Université catholique de Louvain, Belgium.\\[2pt]
\small (2) Program in Applied and Computational Mathematics,\\ \small  Princeton University, NJ, USA.\\[2pt]
$^*${\footnotesize Corresponding author: nicolas.boumal@uclouvain.be}}

\date{\small Compiled on \today}

\maketitle

\begin{abstract}
{Synchronization of rotations is the problem of estimating a set of rotations $R_i\in\SOn, i = 1\ldots N$ based on noisy measurements of relative rotations $R_i^{}R_j\transpose$. This fundamental problem has found many recent applications, most importantly in structural biology. We provide a framework to study synchronization as estimation on Riemannian manifolds for arbitrary $n$ under a large family of noise models. The noise models we address\remove{ essentially} encompass zero-mean isotropic noise, and we develop tools for Gaussian-like as well as heavy-tail types of noise in particular. As a main contribution, we \change{establish formulas for}{derive} the \remove{Fisher information matrix and associated }Cramér-Rao bounds of synchronization\add{, that is, lower-bounds on the variance of unbiased estimators}. We find that these bounds are structured by the pseudoinverse of the measurement graph Laplacian, where edge weights are proportional to measurement quality. We leverage this to provide interpretation \add{in terms of random walks} and visualization tools for these bounds in both the anchored and anchor-free scenarios. Similar bounds previously established were limited to rotations in the plane and Gaussian-like noise.}

{Synchronization of rotations, estimation on manifolds, estimation on graphs, graph Laplacian, Fisher information, Cramér-Rao bounds, distributions on the rotation group, Langevin.}
\\
2000 Math Subject Classification: 
62F99, 
94C15, 
22C05, 
05C12, 
\end{abstract}

\input{intro}

\input{synchronization}

\input{paramspacegeometry}

\input{measuresson}

\input{fimsynch}
\input{crbsynch}

\input{analysis}

\input{conclusions}

\section*{Acknowledgments}
We thank the anonymous referees for their insightful comments. This work originated and was partly conducted during visits of NB at PACM, Princeton University. This paper presents research results of the Belgian Network DYSCO (Dynamical Systems, Control, and Optimization), funded by the Interuniversity Attraction Poles Programme, initiated by the Belgian State, Science Policy Office. NB is an FNRS research fellow (Aspirant). A.\,Singer was partially supported by Award Number DMS-0914892 from the NSF, by Award Number FA9550-12-1-0317 from AFOSR, by Award Number R01GM090200 from the National Institute of General Medical Sciences, by the Alfred P. Sloan Foundation and by Award Number LTR DTD 06-05-2012 from the Simons Foundation.

\appendix

\input{densitynorm}

\input{curvatureterms}
\input{GijGjiproof}

\input{Pklproof}

\bibliographystyle{plainnat}
\bibliography{synchrotbiblio}

\end{document}

%% file: preamble.tex
\usepackage[ansinew]{inputenc}
\usepackage{graphicx}
\usepackage[lmargin=3cm, rmargin=3cm,bottom=3.15cm,top=3cm]{geometry}
\usepackage{amsfonts}
\usepackage{amsmath}
\usepackage{amsthm}
\usepackage{color}
\usepackage{amssymb}
\usepackage{url}
\usepackage{bm} 
\usepackage{dsfont} 

\usepackage{sectsty}
\makeatletter\def\@seccntformat#1{\protect\makebox[0pt][r]{\csname the#1\endcsname\hspace{12pt}}}\makeatother

\usepackage{mathrsfs} 

\numberwithin{equation}{section}

\newcommand{\transpose}{^\top\! }

\newcommand{\inner}[2]{\left\langle{#1},{#2}\right\rangle}
\newcommand{\innersmall}[2]{\langle{#1},{#2}\rangle}
\newcommand{\innerbig}[2]{\big\langle{#1},{#2}\big\rangle}
\newcommand{\MSE}{\mathrm{MSE}}

\newcommand{\trace}{\mathrm{trace}}
\newcommand{\skeww}[1]{\operatorname{skew}\!\left( #1 \right)}
\newcommand{\PP}{\mathrm{P}}
\newcommand{\Exp}{\mathrm{Exp}}
\newcommand{\Log}{\mathrm{Log}}
\newcommand{\expect}{\mathbb{E}}
\newcommand{\expectt}[1]{\mathbb{E}\left\{{#1}\right\}}

\newcommand{\T}{\mathrm{T}}
\newcommand{\HH}{\mathrm{H}}
\newcommand{\VV}{\mathrm{V}}

\newcommand{\M}{\mathcal{M}}
\newcommand{\p}{\mathcal{P}}
\newcommand{\SOn}{{\mathrm{SO}(n)}}
\newcommand{\On}{{\mathrm{O}(n)}}
\newcommand{\son}{{\mathfrak{so}(n)}}
\newcommand{\SOt}{{\mathrm{SO}(3)}}

\newcommand{\Rnn}{{\mathbb{R}^{n\times n}}}
\newcommand{\Rn}{{\mathbb{R}^n}}

\newcommand{\grad}{\mathrm{grad}}

\newcommand{\sign}{\mathrm{sign}}
\newcommand{\diag}{\mathrm{diag}}
\newcommand{\D}{\mathrm{D}}

\newcommand{\frobnorm}[1]{\left\|{#1}\right\|}
\newcommand{\smallfrobnorm}[1]{\big\|{#1}\big\|}
\newcommand{\norm}[1]{\left\|{#1}\right\|}
\newcommand{\sqnorm}[1]{\left\|{#1}\right\|^2}
\newcommand{\sqfrobnorm}[1]{\frobnorm{#1}^2}

\newcommand{\dmu}{\mathrm{d}\mu}

\newcommand{\SHORTEN}[1]{}

\newcommand{\floor}[1]{\lfloor #1 \rfloor}
\newcommand{\uniform}{\mathrm{Uni}}
\newcommand{\langevin}{\mathrm{Lang}}
\newcommand{\langout}{\mathrm{LangUni}}
\newcommand{\Zeq}{{Z_\mathrm{eq}}}
\newcommand{\feq}{{f_\mathrm{eq}}}

\newcommand{\bfR}{\mathbf{R}}
\newcommand{\bfxi}{{\bm{\xi}}}
\newcommand{\bfX}{{\mathbf{X}}}
\newcommand{\bfeta}{{\bm{\eta}}}

\newcommand{\bfOmega}{\boldsymbol{\Omega}}

\newcommand{\ddiag}{\mathrm{ddiag}}
\newcommand{\dist}{\mathrm{dist}}

\newtheorem{theorem}{Theorem}[section]
\newtheorem{lemma}[theorem]{Lemma}
\newtheorem{assumption}{Assumption}
\newtheorem{example}{Example}
\newtheorem{definition}{Definition}
\newtheorem{remark}{Remark}

%% file: intro.tex
\section{Introduction}
\label{sec:intro}

Synchronization of rotations is the problem of estimating rotation matrices $R_1, \ldots, R_N$ from noisy measurements of relative rotations $R_i^{}R_j\transpose$. The set of available measurements gives rise to a graph structure, where the $N$ nodes correspond to the rotations $R_i$ and an edge is present between two nodes $i$ and $j$ if a measurement of $R_i^{}R_j\transpose$ is given. Depending on the application, some rotations may be known in advance or not. The known rotations, if any, are called \emph{anchors}. In the absence of anchors, it is only possible to recover the rotations up to a global rotation, since the measurements only reveal relative information.

Motivated by the pervasiveness of synchronization of rotations in applications, we propose a derivation and analysis of Cramér-Rao bounds for this estimation problem. Our results hold for rotations in the \emph{special orthogonal group}
\begin{align}
	\SOn & = \{R \in \Rnn : R\transpose R = I \textrm{ and } \det R = 1\}
	\label{eq:SOn}
\end{align}
for arbitrary $n$ and for a large family of practically useful noise models.

Synchronization of rotations appears naturally in a number of important applications. Tron and Vidal for example consider a network of cameras~\citep{tron2009distributed}. Each camera has a certain position in $\mathbb{R}^3$ and orientation in $\SOt$. For some pairs of cameras, a calibration procedure produces a noisy measurement of relative position and relative orientation. The task of using all relative orientation measurements simultaneously to estimate the configuration of the individual cameras is a synchronization problem. Cucuringu et al.\ address sensor network localization based on inter-node distance measurements~\citep{cucuringu2011sensor}. In their approach, they decompose the network in small, overlapping, rigid patches. Each patch is easily embedded in space owing to its rigidity, but the individual embeddings are noisy. These embeddings are then aggregated by aligning overlapping patches. For each pair of such patches, a measurement of relative orientation is produced. Synchronization permits the use all of these measurements simultaneously to prevent error propagation. In related work, a similar approach is applied to the molecule problem~\citep{cucuringu2011eigenvector}. Tzeneva et al.\ apply synchronization to the construction of 3D models of objects based on scans of the objects under various unknown orientations~\citep{tzvetelina}. Singer and Shkolnisky study cryo-EM imaging~\citep{singer2011eigen}. In this problem, the aim is to produce a 3D model of a macro-molecule based on many projections (pictures) of the molecule under various random and unknown orientations. A procedure specific to the cryo-EM imaging technique helps estimating the relative orientation between pairs of projections, but this process is very noisy. In fact, most measurements are outliers. The task is to use these noisy measurements of relative orientations of images to recover the true orientations under which the images were acquired. This naturally falls into the scope of synchronization of rotations, and calls for very robust algorithms. \add{More recently, Sonday \protect{et al.\ use} synchronization as a means to compute rotationally invariant distances between snapshots of trajectories of dynamical systems, as an important preprocessing stage before dimensionality reduction~{\citep{SondaySingerKevrekidis2013}}. In a different setting, Yu applies synchronization of in-plane rotations (under the name of angular embedding) as a means to rank objects based on pairwise relative ranking measurements~{\citep{Yu2009luminance,Yu2012angularembedding}}. This approach is in contrast with existing techniques which realize the embedding on the real line, but appears to provide unprecedented robustness.}

\subsection{Previous work}

Tron and Vidal, in their work about camera calibration, develop distributed algorithms based on consensus on manifolds to solve synchronization on $\mathbb{R}^3\rtimes \SOt$~\citep{tron2009distributed}. Singer studies synchronization of phases, that is, rotations in the plane, and reflects upon the generic nature of synchronization as the task of estimating group elements $g_1, \ldots, g_N$ based on measurements of their ratios $g_i^{}g_j^{-1}$~\citep{singer2010angular}. \change{In that work, the author focuses on synchronization in the presence of many outliers in the measurements and establishes that synchronization is surprisingly robust against such outliers. Bounds are derived based on information-theoretic arguments to establish how many measurements need to be accurate for synchronization to be possible. A fast algorithm based on eigenvector computations is shown to obtain good solutions. In particular, it performs better than random estimation as soon as the information-theoretic threshold is reached.}{In that work, the author focuses on synchronization in the presence of many outliers. A fast algorithm based on eigenvector computations is proposed, corresponding to a relaxation of an otherwise untractable optimization problem. The eigenvector method is shown to be remarkably robust. In particular, it is established that a large fraction of measurements can be outliers while still guaranteeing the algorithm performs better than random in expectation.} In further work, Bandeira et al.\ \change{go deeper in the analysis and generalize the bounds and}{derive Cheeger-type inequalities for synchronization on the orthogonal group under adversarial noise and generalize} the eigenvector method to rotations in $\Rn$~\citep{bandeira2012cheeger}. Hartley et al.\ develop a robust algorithm based on a Weiszfeld iteration to compute L1-means on $\SOn$, and extend the algorithm to perform robust synchronization of rotations~\citep{hartley2011l1}. \add{In a more recent paper~{\cite{hartley2013rotation}}, some of these authors and others address a broad class of rotation averaging problems, with a specific outlook for characterizations of the existence and uniqueness of global optimizers of the related optimization problems. Synchronization is addressed too under the name of \emph{multiple rotation averaging}. }\add{Wang and Singer propose a robust algorithm for synchronization called LUD for \emph{least unsquared deviation}~{\citep{wang2012LUD}}. It is based on a convex relaxation of an L1 formulation of the synchronization problem and comes with exact and stable recovery guarantees under a large set of scenarios. }Russel et al.\ develop a decentralized algorithm for synchronization on the group of translations $\Rn$~\citep{russell2011optimal}. Barooah and Hespanha study the covariance of the BLUE estimator for synchronization on $\Rn$ with anchors. This covariance coincides with the Cramér-Rao bound (CRB) under Gaussian noise. They give interpretations of the covariance in terms of the resistance distance on the measurement graph~\citep{barooah2007estimation}. Howard et al.\ study synchronization on the group of translations $\Rn$ and on the group of phases $\mathrm{SO}(2)$~\citep{howard2010estimation}. They establish CRB's for synchronization in the presence of Gaussian-like noise on these groups and provide decentralized algorithms to solve synchronization. Their derivation of the CRB's seems to rely heavily on the commutativity (and thus flatness) of $\mathbb{R}^n$ and $\mathrm{SO}(2)$, and hence does not apply to synchronization on $\SOn$ in general. Furthermore, they only analyze Gaussian-like noise.

CRB's are a classical tool in estimation theory~\citep{rao1945information} that provide a lower-bound on the variance of any unbiased estimator for an estimation problem, based on the measurements' distribution. The classical results focus on estimation problems for which the sought parameter belongs to a Euclidean space. In the context of synchronization, this is not sufficient, since the parameters we seek belong to the manifold of rotations $\SOn$. Important work by Smith~\citep{smith2005covariance} as well as by Xavier and Barroso~\citep{xavier2005intrinsic} extends the theory of CRB's to the realm of manifolds in a practical way. Furthermore, in the absence of anchors, \change{synchronization can only be solved up to a global rotation. This}{the global rotation} ambiguity\remove{, as we shall see,} gives rise to a singular Fisher information matrix (FIM). As the CRB's are usually expressed in terms of the inverse of the FIM, this is problematic. Xavier and Barroso provide a nice geometric interpretation in terms of estimation on quotient manifolds of the well-known fact that one may use the pseudoinverse of the FIM in such situations~\citep{xavier2004singularfisher}. It is then apparent that to establish CRB's for synchronization, one needs to either impose anchors, leading to a submanifold geometry, or work on the quotient manifold. In~\citep{crbsubquot}, an attempt is made at providing a unified framework to build such CRB's. We use these tools in the present work.

Other authors have established CRB's for sensor network localization (SNL) and synchronization problems. Ash and Moses study SNL based on inter-agent distance measurements, and notably give an interpretation of the CRB in the absence of anchors~\citep{ash2007relative}. Chang and Sahai tackle the same problem~\citep{chang2006cramer}. As we mentioned earlier, Howard et al.\ derive CRB's for synchronization on $\Rn$ and $\mathrm{SO}(2)$~\citep{howard2010estimation}. CRB's for synchronization on $\Rn$ are re-derived as an example in~\citep{crbsubquot}. A remarkable fact is that, for all these problems of estimation on graphs, the pseudoinverse of the graph Laplacian plays a fundamental role in the CRB---although not all authors explicitly reflect on this. As we shall see, this special structure is rich in interpretations, many of which exceed the context of synchronization of rotations specifically.

\subsection{Contributions and outline}

In this work, we state the problem of synchronization of rotations in $\SOn$ for arbitrary $n$ as an estimation problem on a manifold---Section~\ref{sec:synchro}. We describe the geometry of this manifold both for anchored synchronization (giving rise to a submanifold geometry) and for anchor-free synchronization (giving rise to a quotient geometry)---Section~\ref{sec:geometry}. Among other things, this paves the way for maximum-likelihood estimation using optimization on manifolds~\citep{AMS08}, which is the focus of other work~\citep{boumal2013MLE}.

We describe a family of noise models (probability density functions) on $\SOn$ that fulfill a few assumptions---Section~\ref{sec:measuresson}. We show that this family is both useful for applications (it essentially contains zero-mean, isotropic noise models) and practical to work with (the expectations one is lead to compute via integrals on $\SOn$ are easily converted into classical integrals on $\Rn$). In particular, this family includes a kind of heavy-tailed distribution on $\SOn$ that appears to be new. We describe this distribution and we are convinced it will prove useful for other estimation problems on $\SOn$ with outliers.

In Section~\ref{sec:fimsynch}, we derive the Fisher information matrix (FIM) for synchronization and establish it is structured by the Laplacian of the measurement graph, where edge weights are proportional to the quality of their respective measurements. The FIM plays a central role in the Cramér-Rao bounds (CRB) we establish for anchored and anchor-free synchronization---Section~\ref{sec:crbsynch}.\add{ The main tools used to that effect are intrinsic versions of the CRB's, as developed in~{\citep{smith2005covariance}} and, in a formulation more directly useful to our setting, in~{\citep{crbsubquot}}.} The CRB's are structured by the pseudoinverse of the Laplacian of the measurement graph. We derive clear interpretations of these bounds in both cases, and note they differ significantly.

As a main result for anchored synchronization, we show that for any unbiased estimator $\hat R_i$ of the rotation $R_i$, asymptotically for small errors,
\begin{align}
	\expectt{\dist^2(R_i, \hat R_i)} \geq d^2\, (\mathcal{L}_A^\dagger)_{ii},
	\label{eq:intro1}
\end{align}
where $\dist(R_i, \hat R_i) = \|\log(R_i\transpose\hat R_i^{})\|_\mathrm{F}$ is the geodesic distance on $\SOn$, $d = n(n-1)/2$, $\mathcal{L}_A$ is the Laplacian of the weighted measurement graph with rows and columns corresponding to anchors set to zero and $\dagger$ denotes the Moore-Penrose pseudoinverse---see~\eqref{eq:erroranchored}. The better a measurement is, the larger the weight on the associated edge is---see~\eqref{eq:infoweight}. This bound holds in a small-error regime under the assumption that noise on different measurements is independent, that the measurements are isotropically distributed around the true relative rotations and that there is at least one anchor in each connected component of the graph. \add{The right-hand side of this inequality is zero if node $i$ is an anchor, and is small if node $i$ is strongly connected to anchors. More precisely, it is proportional to the ratio between the average number of times a random walker starting at node $i$ will be at node $i$ before hitting an anchored node and the total amount of information available in measurements involving node $i$.}

As a main result for anchor-free synchronization, we show that for any unbiased estimator $\hat R_i^{}\hat R_j\transpose$ of the relative rotation $R_i^{}R_j\transpose$, asymptotically for small errors,
\begin{align}
	\expectt{\dist^2(R_i^{}R_j\transpose, \hat R_i^{}\hat R_j\transpose)} \geq d^2\,(e_i-e_j)\transpose \mathcal{L}^\dagger (e_i-e_j),
	\label{eq:intro2}
\end{align}
where $\mathcal{L}$ is the Laplacian of the weighted measurement graph and $e_i$ is the $i^\textrm{th}$ column of the $N\times N$ identity matrix---see~\eqref{eq:erroranchorfree}. This bound holds in a small-error regime under the assumption that noise on different measurements is independent, that the measurements are isotropically distributed around the true relative rotations and that the measurement graph is connected.\add{ The right-hand side of this inequality is proportional to the squared Euclidean commute time distance (ECTD)~{\citep{saerens2004principal}} on the weighted graph. It measures how strongly nodes $i$ and $j$ are connected. More explicitly, it is proportional to the average time a random walker starting at node $i$ walks before hitting node $j$ and then node $i$ again.}

Section~\ref{sec:analysis} hosts a few comments on the CRB's. In particular, evidence is presented that the CRB might be achievable, a PCA-like visualization tool is detailed, a link with the Fiedler value of the graph is described and the robustness of synchronization versus outliers is confirmed, via arguments that differ from those in~\citep{singer2010angular}.

\add{\protect\paragraph{Notation.} $I_n$ is the identity matrix of size $n$. When it is clear from the context, the subscript is omitted. $\On = \{R\in\Rnn : R\transpose R = I\}$ is the orthogonal group. We synchronize $N$ rotations in $\SOn$~\eqref{eq:SOn}. Anchors (if any) are indexed in $A \subset \{1,\ldots,N\}$. $d = \dim \SOn = n(n-1)/2$. $\otimes$ denotes the Kronecker product. $\frobnorm{A}^2 = \trace(A\transpose A)$ is the squared Frobenius norm. $\skeww{A} = (A-A\transpose)/2$ is the skew-symmetric part of $A$. $A^\dagger$ is the Moore-Penrose pseudoinverse of $A$. Bold letters $\bfR$ denote tuples $(R_1, \ldots, R_N)$.}

%% file: synchronization.tex
\section{Synchronization of rotations}
\label{sec:synchro}

Synchronization of rotations is the problem of estimating a set of rotations $R_1, \ldots, R_N \in \SOn$ from noisy measurements of some of the relative rotations $R_i^{}R_j\transpose$. In this section, we model synchronization as an estimation problem on a manifold, since the set of rotations $\SOn$ is a Lie group, i.e., a group and a manifold at the same time.

In our estimation problem, the target quantities (the parameters) are the rotation matrices $R_1, \ldots, R_N \in \SOn$. The natural parameter space is thus:
\begin{align}
	\p = \SOn\times \cdots \times \SOn \qquad (N \textrm{ copies}).
	\label{eq:P}
\end{align}
Let $[N] \triangleq \{1,\ldots,N\}$. Consider a set $\mathcal{E} \subset [N]\times [N]$ such that $(i, j)\in\mathcal{E} \Rightarrow i\neq j$ and $(j,i)\in\mathcal{E}$. This set defines an undirected graph over $N$ nodes,
\begin{align}
	G & = ([N], \mathcal{E}) & (\textrm{the measurement graph}).
\end{align}
For each edge $(i, j) \in \mathcal{E}$, $i<j$, we have a measurement $H_{ij} \in \SOn$ of the form
\begin{align}
	H_{ij} = Z_{ij}^{}R_i^{}R_j\transpose,
	\label{eq:Hij}
\end{align}
where $Z_{ij}$ is a random variable distributed over $\SOn$ following a probability density function (pdf) $f_{ij} : \SOn \to \mathbb{R}^+$, with respect to the Haar measure $\mu$ on $\SOn$---see Section~\ref{sec:measuresson}. For example, when $Z_{ij}$ is deterministically equal to the identity matrix $I$, the measurement is perfect, whereas when $Z_{ij}$ is uniformly distributed over (the compact set) $\SOn$, the measurement contains no information. We say that the measurement is \emph{unbiased}, or that the noise has \emph{zero-mean}, if the mean value of $Z_{ij}$ is the identity matrix---a notion we make precise in Section~\ref{sec:measuresson}. We also say that noise is \emph{isotropic} if its probability density function is only a function of distance to the identity. Different notions of distance on $\SOn$ yield different notions of isotropy. In Section~\ref{sec:measuresson} we give a few examples of useful zero-mean, isotropic distributions on $\SOn$.

Pairs $(i, j)$ and $(j, i)$ in $\mathcal{E}$ refer to the same measurement. By symmetry, for $i<j$, we define $H_{ji}^{} = Z_{ji}^{}R_j^{}R_i\transpose = H_{ij}\transpose$ and the random variable $Z_{ji}$ and its density $f_{ji}$ are defined accordingly in terms of $f_{ij}$ and $Z_{ij}$. In particular,
\begin{align}
	Z_{ji} & = R_j^{}R_i\transpose Z_{ij}\transpose R_i^{} R_j\transpose, \textrm{ and } & f_{ij}(Z_{ij}) & = f_{ji}(Z_{ji}).
	\label{eq:ijfixesji}
\end{align}
The pdf's $f_{ij}$ and $f_{ji}$ are linked as such because the Haar measure $\mu$ is invariant under the change of variable relating $Z_{ij}$ and $Z_{ji}$.

In this work, we restrict our attention to noise models that fulfill the three following assumptions:
\begin{assumption}[smoothness and support]\label{assu:smoothpositive}
	Each pdf $f_{ij} : \SOn \to \mathbb{R}^+ = (0, +\infty)$ is a smooth, positive function.
\end{assumption}
\begin{assumption}[independence]\label{assu:independence}
The $Z_{ij}$'s associated to different edges of the measurement graph are independent random variables. That is, if $(i,j)\neq(p,q)$ and $(i,j)\neq(q,p)$,  $Z_{ij}$ and $Z_{pq}$ are independent.
\end{assumption}
\begin{assumption}[invariance]\label{assu:invariant}
	Each pdf $f_{ij}$ is invariant under orthogonal conjugation, that is, $\forall Z\in\SOn, \forall Q \in \On, f_{ij}(QZQ\transpose) = f_{ij}(Z).$ We say $f_{ij}$ is a \emph{spectral function}, since it only depends on the eigenvalues of its argument. The eigenvalues of matrices in $\mathrm{SO}(2k)$ have the form $e^{\pm i\theta_1}, \ldots, e^{\pm i\theta_k}$, with $0 \leq \theta_1, \ldots, \theta_k \leq \pi$. The eigenvalues of matrices in $\mathrm{SO}(2k+1)$ have an additional eigenvalue 1.
\end{assumption}
Assumption~\ref{assu:smoothpositive} is satisfied for all the noise models we consider; it could be relaxed to some extent but would make some of the proofs more technical. Assumption~\ref{assu:independence} \add{is admittedly a strong restriction but }is necessary to make the joint pdf of the whole estimation problem easy to derive, leading to an easy expression for the log-likelihood function. As we will see in Section~\ref{sec:fimsynch}, it is also at the heart of the nice Laplacian structure of the Fisher information matrix. Assumption~\ref{assu:invariant} is a technical condition that will prove useful in many respects. One of them is the observation that pdf's which obey Assumption~\ref{assu:invariant} are easy to integrate over $\SOn$. We expand on this in Section~\ref{sec:measuresson}, where we also show that a large family of interesting pdf's satisfy these assumptions, namely, zero-mean isotropic distributions.

Under Assumption~\ref{assu:independence}, the \emph{log-likelihood} of an estimator $\hat \bfR = (\hat R_1, \ldots, \hat R_N) \in \p$, given the measurements $H_{ij}$, is given by:
\begin{align}
	L(\hat \bfR) & = \frac{1}{2} \sum_{(i,j)\in\mathcal{E}} \log f_{ij}(H_{ij}^{}\hat R_j^{} \hat R_i\transpose) = \frac{1}{2} \sum_{i = 1}^N \sum_{j\in V_i} \log f_{ij}(H_{ij}^{}\hat R_j^{} \hat R_i\transpose),
	\label{eq:L}
\end{align}
where $V_i \subset [N]$ is the set of neighbors of node $i$, i.e., $j\in V_i \Leftrightarrow (i, j) \in \mathcal{E}$. The coefficient $1/2$ reflects the fact that measurements $(i,j)$ and $(j,i)$ give the same information and are deterministically linked. Under Assumption~\ref{assu:smoothpositive}, $L$ is a smooth function on the smooth manifold $\p$.


The log-likelihood function is invariant under a global rotation. Indeed,
\begin{align}
	\forall \hat \bfR \in \p,\ \forall Q \in \SOn,\quad L(\hat \bfR Q) = L(\hat \bfR ),
	\label{eq:invariance}
\end{align}
where $\hat \bfR Q$ denotes $(\hat R_1Q, \ldots, \hat R_NQ) \in \p$. This invariance encodes the fact that all sets of rotations of the form $\hat \bfR Q$ yield the same distribution of the measurements $H_{ij}$, and are hence \change{indistinguishable, i.e., rotations can only be recovered up to a global rotation}{equally likely estimators}.

To resolve the ambiguity, one can follow at least two courses of action. One is to include additional constraints, most naturally in the form of anchors, i.e., assume some of the rotations are known\footnote{If we only know that $R_i$ is close to some matrix $\bar R$, and not necessarily equal to it, we may add a phony node $R_{N+1}$ anchored at $\bar R$, and link that node and $R_i$ with a high confidence measure $H_{i,N+1}=I_n$. This makes it possible to have ``soft anchors''.}. The other is to acknowledge the invariance by working on the associated quotient space.

Following the first path, the parameter space becomes a Riemannian submanifold of $\p$. Following the second path, the parameter space becomes a Riemannian quotient manifold of $\p$. In the next section, we describe the geometry of both. \remove{In Section~{\ref{sec:measuresson}}, we review useful tools to compute integrals on $\SOn$ and describe probability density functions that are both useful to model practical problems and fulfill our assumptions. In Section~{\ref{sec:fimsynch}} we use the expression of the log-likelihood function~{\eqref{eq:L}} to derive the Fisher information matrix for synchronization, which plays a central role in deriving the Cramér-Rao bounds for synchronization with and without anchors---see Section~{\ref{sec:crbsynch}}.}

\begin{remark}[A word about other noise models]

We show that measurements of the form $H_{ij}^{} = Z_{ij,1}^{}R_i^{}R_j\transpose Z_{ij,2}^{}$, with $Z_{ij,1}$ and $Z_{ij,2}$ two \remove{independent }random rotations with pdf's satisfying Assumptions 1--3, satisfy the noise model considered in the present work. In doing so, we use some material from Section~\ref{sec:measuresson}. For notational convenience, let us consider $H = Z_1RZ_2$, with $Z_1, Z_2$ two \remove{independent }random rotations with pdf's $f_1, f_2$ satisfying Assumptions 1--3, $R \in \SOn$ fixed. Then, the pdf of $H$ is the function $h : \SOn \to \mathbb{R}^+$ given by (essentially) the convolution of $f_1$ and $f_2$ on $\SOn$:
\begin{align}
	h(H) & = \int_\SOn f_1(Z) f_2(R\transpose Z\transpose H)\, \dmu(Z) 
		   = \int_\SOn f_1(Z) f_2(Z\transpose H R\transpose) \,\dmu(Z),
\end{align}
where we used that $f_2$ is spectral: $f_2(R\transpose Z\transpose H) = f_2(R R\transpose Z\transpose H R\transpose)$. Let $\Zeq$ be a random rotation with smooth pdf $\feq$. We will shape $\feq$ such that the random rotation $\Zeq R$ has the same distribution as $H$. This condition can be written as follows: for all measurable subsets $S\subset \SOn$,
\begin{align}
	\int_S h(Z) \,\dmu(Z) = \int_{SR\transpose} \feq(Z) \,\dmu(Z) = \int_S \feq(ZR\transpose) \,\dmu(Z),
\end{align}
where, going from the second to the third integral, we used the change of variable $Z:=ZR\transpose$ and the bi-invariance of the Haar measure $\mu$. In words: for all $S$, the probability that $H$ belongs to $S$ must be the same as the probability that $\Zeq R$ belongs to $S$. This must hold for all $S$, hence $\feq(\Zeq R\transpose) = h(\Zeq)$, or equivalently:
\begin{align}
	\feq(\Zeq) = h(\Zeq R) = \int_\SOn f_1(Z) f_2(Z\transpose \Zeq) \,\dmu(Z).
\end{align}
This uniquely defines the pdf of $\Zeq$. It remains to show that $\feq$ is a spectral function. For all $Q\in\On$,
\begin{align}
	\feq(Q \Zeq Q\transpose)
		& = \int_\SOn f_1(Z) f_2(Z\transpose Q \Zeq Q\transpose) \,\dmu(Z) \\
		\textrm{($f_2$ is spectral)} \quad & = \int_\SOn f_1(Z) f_2(Q\transpose Z\transpose Q \Zeq) \,\dmu(Z) \\
		\textrm{(change of variable: $Z:=QZQ\transpose$)} \quad & = \int_\SOn f_1(QZQ\transpose ) f_2(Z\transpose \Zeq) \,\dmu(Z) \\
		\textrm{($f_1$ is spectral)} \quad & = \int_\SOn f_1(Z) f_2(Z\transpose \Zeq) \,\dmu(Z) = \feq(\Zeq).
\end{align}
Hence, the noise model $H_{ij}^{} = Z_{ij,1}^{}R_i^{}R_j\transpose Z_{ij,2}^{}$ can be replaced with the model $H_{ij}^{} = Z_{ij,\mathrm{eq}}^{}R_i^{}R_j\transpose$ and the pdf of $Z_{ij,\mathrm{eq}}$ is such that it falls within the scope of the present work.

In particular, if $f_1$ is a point mass at the identity, so that $H = RZ_2$ (noise multiplying the relative rotation on the right rather than on the left), $\feq = f_2$, so that it does not matter whether we consider $H_{ij}^{} = Z_{ij}^{}R_i^{}R_j\transpose$ or $H_{ij}^{} = R_i^{}R_j\transpose Z_{ij}^{}$: they have the same distribution.

\end{remark}

%% file: paramspacegeometry.tex
\section{Geometry of the parameter spaces}
\label{sec:geometry}

\add{This section defines the notions of distance involved in the CRB's~\eqref{eq:intro1} and~\eqref{eq:intro2}. It introduces tools to obtain the gradient involved in the FIM~\eqref{eq:gradiLhatR} as well as the Riemannian $\Log$ map necessary to define notions of error and bias for estimators on manifolds. This is achieved by defining the parameter spaces for both the anchored and the anchor-free scenario and endowing them with a proper Riemannian structure.}

We start with a quick reminder of the geometry of $\SOn$. We then go on to describe the parameter spaces for the anchored and the anchor-free cases of synchronization. It is assumed that the reader is familiar with standard concepts from Riemannian geometry~\citep{AMS08,boothby1986introduction,oneill}. \add{For readers less comfortable with Riemannian geometry, it may be helpful to consider anchored synchronization only at first. The distinction between both cases is clearly delineated in the remainder of the paper, most of which remains relevant even if only the anchored case is considered.}

The group of rotations $\SOn$~\eqref{eq:SOn} is a connected, compact Lie group of dimension $d = n(n-1)/2$. \add{Being a Lie group, it is also a manifold and thus admits a tangent space $\T_Q\SOn$ at each point $Q$. The tangent space at the identity plays a special role. It is known as the \emph{Lie algebra} of $\SOn$ and is the set of skew-symmetric matrices:}
\begin{align}
	\T_I\SOn = \son & \triangleq \{ \Omega \in \Rnn : \Omega + \Omega\transpose = 0 \}.
\end{align}
\add{The other tangent spaces are easily obtained from $\son$:}
\begin{align}
	\T_Q\SOn = Q\son = \{Q\Omega : \Omega \in \son\}.
\end{align}
\change{As usual, we consider $\SOn$ as a Riemannian submanifold of the general linear group, and hence endow it with the usual metric}{We endow $\SOn$ with the usual Riemannian metric by defining the following inner product on all tangent spaces}:
\begin{align}
	\inner{Q\Omega_1}{Q\Omega_2}_Q & = \trace(\Omega_1\transpose\Omega_2^{}), & \sqnorm{Q\Omega}_Q = \inner{Q\Omega}{Q\Omega}_Q = \sqfrobnorm{\Omega}.
\end{align}
For better readability, we often omit the subscripts $Q$. The orthogonal projector from the embedding space $\Rnn$ onto the tangent space $\T_Q\SOn$ is:
\begin{align}
	\PP_Q(H) = Q\skeww{Q\transpose H}, \textrm{ with } \skeww{A} \triangleq (A-A\transpose)/2.
	\label{eq:projSOn}
\end{align}
\add{It plays an important role in the computation of gradients of functions on $\SOn$, which will come up in deriving the FIM. }The exponential map and the logarithmic map accept simple expressions in terms of matrix exponential and logarithm:
\begin{align}
	\Exp_Q & \colon \T_Q\SOn \to \SOn, &
	\Log_Q & \colon \SOn \to \T_Q\SOn \\
	\Exp_Q&(Q\Omega) = Q\exp(\Omega), &
	\Log_{Q_1}&(Q_2) = Q_1\log(Q_1\transpose Q_2^{}). \label{eq:ExpLogSOn}
\end{align}
The mapping $t \mapsto \Exp_Q(tQ\Omega)$ defines a geodesic curve on $\SOn$, passing through $Q$ with velocity $Q\Omega$ at time $t=0$. Geodesic curves have zero acceleration and may be considered as the equivalent of straight lines on manifolds. The logarithmic map $\Log_{Q}$ is (locally) the inverse of the exponential map $\Exp_{Q}$. \add{In the context of an estimation problem, $\Log_Q(\hat Q)$ represents the estimation error of $\hat Q$ for the parameter $Q$, that is, it is a notion of difference between $Q$ and $\hat Q$.} The geodesic (or Riemannian) distance on $\SOn$ is the length of the shortest path (the geodesic arc) joining two points:
\begin{align}
	\dist(Q_1, Q_2) = \norm{\Log_{Q_1}(Q_2)}_{Q_1} = \smallfrobnorm{\log(Q_1\transpose Q_2^{})}.
	\label{eq:distSOn}
\end{align}
In particular, for rotations in the plane ($n = 2$) and in space ($n = 3$), the geodesic distance between $Q_1$ and $Q_2$ is $\sqrt{2}\theta$, where $\theta\in[0,\pi]$ is the angle by which $Q_1\transpose Q_2^{}$ rotates.

Let $\tilde f : \Rnn \to \mathbb{R}$ be a differentiable function, and let $f = \tilde f|_{\SOn}$ be its restriction to $\SOn$. The gradient of $f$ is a tangent vector field to $\SOn$ uniquely defined by:
\begin{align}
	\inner{\grad f(Q)}{Q\Omega} = \D f(Q)[Q\Omega] \quad \forall \Omega \in \son,
\end{align}
with $\grad f(Q) \in \T_Q\SOn$ and $\D f(Q)[Q\Omega]$ the directional derivative of $f$ at $Q$ along $Q\Omega$. Let $\nabla \tilde f(Q)$ be the usual gradient of $\tilde f$ in $\Rnn$. Then, the gradient of $f$ is easily computed as~\citep[eq.\,(3.37)]{AMS08}:
\begin{align}
	\grad f(Q) = \PP_Q(\nabla \tilde f(Q)).
\end{align}
In the sequel, we often write $\nabla f$ to denote the gradient of $f$ seen as a function in $\Rnn$, even if it is defined on $\SOn$.

The parent parameter space for synchronization is the product Lie group $\p = \SOn^N$. Its geometry is trivially obtained by element-wise extension of the geometry of $\SOn$ just described. In particular, tangent spaces and the Riemannian metric are given by:
\begin{align}
	\T_{\bfR}\p & = \{ \bfR \bfOmega = (R_1\Omega_1, \ldots, R_N\Omega_N) : \Omega_1, \ldots, \Omega_N \in \son \}, \label{eq:TRP}\\
	\inner{\bfR\bfOmega}{\bfR\bfOmega'}_{\bfR} & = \sum_{i=1}^N \trace(\Omega_i\transpose \Omega_i'). \label{eq:metricTRP}
\end{align}

\subsection{Anchored case}
\label{subsec:anchoredgeometry}

In specific applications, we may know some of the rotation matrices $R_i$. Let $A \subset [N]$ be the set of indices of known rotations, called \emph{anchors}. The associated parameter space
\begin{align}
	\p_A = \{\hat \bfR = (\hat R_1, \ldots, \hat R_N)\in\p : \forall i\in A, \hat R_i = R_i\}
	\label{eq:pA}
\end{align}
is a Riemannian submanifold of $\p$. The tangent space at $\hat \bfR\in \p_A$ is given by:
\begin{align}
	\T_{\hat \bfR}\p_A & = \{\mathbf{H}\in\T_{\hat \bfR}\p : \forall i\in A, H_i = 0 \},
\end{align}
such that the orthogonal projector $\PP_{\hat \bfR} : \T_{\hat \bfR}\p \to \T_{\hat \bfR}\p_A$ simply sets to zero all components of a tangent vector that correspond to anchored rotations. All tools on $\p_A$ (exponential and logarithmic map for example) are inherited in the obvious fashion from $\p$. In particular, the geodesic distance on $\p_A$ is:
\begin{align}
	\dist^2(\hat \bfR, \hat \bfR') & = \sum_{i\notin A} \smallfrobnorm{\log(\hat R_i\transpose \hat R_i')}^2.
	\label{eq:distanchored}
\end{align}

\subsection{Anchor-free case}
\label{subsec:nonanchoredgeometry}

When no anchors are provided, \remove{the rotations can only be recovered up to a common rotation, since the measurements only provide information regarding relative rotations. More precisely, }the distribution of the measurements $H_{ij}$~\eqref{eq:Hij} is the same whether the true rotations are $(R_1, \ldots, R_N)$ or $(R_1Q, \ldots, R_NQ)$, regardless of $Q\in\SOn$. Consequently, the measurements contain no information as to which of those sets of rotations is the right one. This leads to the definition of the equivalence relation $\sim$:
\begin{align}
	(R_1, \ldots, R_N) \sim (R_1', \ldots, R_N') \quad \Leftrightarrow \quad
 \exists Q \in\SOn : R_i = R_i'Q \textrm{ for } i = 1,\ldots, N.
	\label{eq:equivrelation}
\end{align}
This equivalence relation partitions $\p$ into equivalence classes, often called \emph{fibers}. The quotient space (the set of equivalence classes)
\begin{align}
	\p_\emptyset \triangleq \p/\sim
	\label{eq:pempty}
\end{align}
is again a smooth manifold (in fact, $\p_\emptyset$ is a \emph{coset manifold} because it results from the quotient of the Lie group $\p$ by a closed subgroup of $\p$~\cite[Prop.\,11.12]{oneill}). The notation reminds us that the set of anchors $A$ is empty. Naturally, the log-likelihood function $L$~\eqref{eq:L} is constant over equivalence classes and hence descends as a well-defined function on $\p_\emptyset$.
Each fiber
\begin{align}
	[\bfR] & = \{ (R_1Q, \ldots, R_NQ) : Q \in \SOn \} \in \p_\emptyset
	\label{eq:equivclass}
\end{align}
is a Riemannian submanifold of the total space $\p$. As such, at each point $\bfR$, the fiber $[\bfR]$ admits a tangent space that is a subspace of $\T_{\bfR}\p$. That tangent space to the fiber is called the \emph{vertical space} at $\bfR$, noted $\VV_{\bfR}$. Vertical vectors point along directions that are parallel to the fibers. Vectors orthogonal, in the sense of the Riemannian metric~\eqref{eq:metricTRP}, to all vertical vectors form the \emph{horizontal space} $\HH_{\bfR} = (\VV_{\bfR})^\perp$, such that the tangent space $\T_{\bfR}\p$ is equal to the direct sum $\VV_{\bfR} \oplus \HH_{\bfR}$. Horizontal vectors are orthogonal to the fibers, hence point toward the other fibers, i.e., the other points on the quotient space $\p_\emptyset$.

Because $\p_\emptyset$ is a coset manifold, the projection
\begin{align}
	\pi : \p \to \p_\emptyset : \bfR \mapsto \pi(\bfR) = [\bfR]
\label{eq:submersion}
\end{align}
is a \emph{submersion}. That is, the restricted differential $\D\pi|_{\HH_{\bfR}}$ is a full-rank linear map between $\HH_\bfR$ and $\T_{[\bfR]}\p_\emptyset$. Practically, this means that the horizontal space $\HH_{\bfR}$ is naturally identified to the (abstract) tangent space $\T_{[\bfR]}\p_\emptyset$. This results in a practical means of representing abstract vectors of $\T_{[\bfR]}\p_\emptyset$ simply as vectors of $\HH_{\bfR}\subset \T_{\bfR}\p$, where $\bfR$ is any arbitrarily chosen member of $[\bfR]$. Each horizontal vector $\xi_{\bfR}$ is unambiguously related to its abstract counterpart $\xi_{[\bfR]}$ in $\T_{[\bfR]}\p_\emptyset$ via
\begin{align}
	\xi_{[\bfR]} = \D\pi(\bfR)[\xi_{\bfR}].
	\label{eq:horizontallift}
\end{align}
The representation $\xi_{\bfR}$ of $\xi_{[\bfR]}$ is called the \emph{horizontal lift} of $\xi_{[\bfR]}$ at $\bfR$, a notion made precise in~\citep[\S\,3.5.8]{AMS08} and depicted for intuition in~\citep[Fig.\,1--2]{xavier2004singularfisher}---these figures are also reproduced in~\citep[Fig.\,2--3]{crbsubquot}.

Consider $\xi_{[\bfR]}$ and $\eta_{[\bfR]}$, two tangent vectors at $[\bfR]$. Let $\xi_\bfR$ and $\eta_\bfR$ be their horizontal lifts at $\bfR \in [\bfR]$ and let $\xi_{\bfR'}$ and $\eta_{\bfR'}$ be their horizontal lifts at $\bfR' \in [\bfR]$. The Riemannian metric on $\p$~\eqref{eq:metricTRP} is such that $\inner{\xi_\bfR}{\eta_\bfR}_\bfR = \inner{\xi_{\bfR'}}{\eta_{\bfR'}}_{\bfR'}$. This motivates us to define the metric
\begin{align}
	\inner{\xi_{[\bfR]}}{\eta_{[\bfR]}}_{[\bfR]} = \inner{\xi_\bfR}{\eta_\bfR}_\bfR
\end{align}
on $\p_\emptyset$, which is then well defined (it does not depend on our choice of $\bfR$ in $[\bfR]$) and turns the restricted differential $\D\pi(\bfR) : \HH_{\bfR} \to \T_{[\bfR]}\p_\emptyset$ into an isometry. This is a Riemannian metric and it is the only such metric such that $\pi$~\eqref{eq:submersion} is a  \emph{Riemannian submersion} from $\p$ to $\p_\emptyset$~\cite[Prop.\,2.28]{gallot2004riemannian}. 
Hence, $\p_\emptyset$ is a \emph{Riemannian quotient manifold} of $\p$.


%

We now describe the vertical and horizontal spaces of $\p$ w.r.t.\ the equivalence relation~\eqref{eq:equivrelation}. Let $\bfR\in\p$ and $Q : \mathbb{R} \to \SOn : t \mapsto Q(t)$ such that $Q$ is smooth and $Q(0) = I$. Then, the derivative $Q'(0) = \Omega$ is some skew-symmetric matrix in $\son$. Since $(R_1Q(t), \ldots, R_NQ(t)) \in [\bfR]$ for all $t$, it follows that $\frac{\mathrm{d}}{\mathrm{d}t} \left. (R_1Q(t), \ldots, R_NQ(t))\right|_{t=0} = (R_1\Omega, \ldots, R_N\Omega)$ is a tangent vector to the fiber $[\bfR]$ at $\bfR$, i.e., it is a vertical vector at $\bfR$. All vertical vectors have such form, hence:
\begin{align}
	\VV_{\bfR} & = \big\{ (R_1\Omega, \ldots, R_N\Omega) : \Omega \in \son \big\}.
	\label{eq:vertical}
\end{align}
A horizontal vector $(R_1\Omega_1, \ldots, R_N\Omega_N) \in \HH_{\bfR}$ is orthogonal to all vertical vectors, i.e., $\forall \Omega\in\son$,
\begin{align}
	0 = \inner{(R_1\Omega_1, \ldots, R_N\Omega_N)}{(R_1\Omega, \ldots, R_N\Omega)} & = \sum_{i=1}^N \trace((R_i\Omega_i)\transpose R_i\Omega) 
= \innerbig{\sum_{i=1}^N \Omega_i}{\Omega}.
\end{align}
Since this is true for all skew-symmetric matrices $\Omega$, we find that the horizontal space is defined as:
\begin{align}
	\HH_{\bfR} & = \big\{ (R_1\Omega_1, \ldots, R_N\Omega_N) : \Omega_1, \ldots, \Omega_N\in \son \textrm{ and } \sum_{i=1}^N \Omega_i = 0\big\}.
	\label{eq:horizontal}
\end{align}
This is not surprising: vertical vectors move all rotations in the same direction, remaining in the same equivalence class, whereas horizontal vectors move away toward other equivalence classes.

We now define the logarithmic map on $\p_\emptyset$. Considering two points $[\bfR], [\hat \bfR] \in \p_\emptyset$, the logarithm $\Log_{[\bfR]}([\hat \bfR])$ is the smallest tangent vector in $\T_{[\bfR]}\p_\emptyset$ that brings us from the first equivalence class to the other through the exponential map. In other words: it is the error vector of {$[\hat\bfR]$} in estimating {$[\bfR]$}. Working with the horizontal lift representation
\begin{align}
	\D\pi(\bfR)|_{\HH_{\bfR}}^{-1}[\Log_{[\bfR]}([\hat \bfR])] & = (R_1\Omega_1, \ldots, R_N\Omega_N) \in \HH_{\bfR},
	\label{eq:loganchored}
\end{align}
the $\Omega_i$'s are skew-symmetric matrices solution of:
\begin{align}
	\min_{\Omega_i \in \son, Q\in\SOn} & \sqfrobnorm{\Omega_1} + \cdots + \sqfrobnorm{\Omega_N}, \\
	\textrm{ such that } & R_i\exp(\Omega_i) = \hat R_i^{} Q, i = 1\ldots N, \textrm{ and} \\
								& \Omega_1 + \cdots + \Omega_N = 0.
\end{align}
The rotation $Q$ sweeps through all members of the equivalence class $[\hat \bfR]$ in search of the one closest to $\bfR$. By substituting $\Omega_i = \log(R_i\transpose \hat R_i^{}Q)$ in the objective function, we find that the objective value as a function of $Q$ is $\sum_{i=1\ldots N} \|\log(R_i\transpose \hat R_i^{}Q)\|_\textrm{F}^2$. Critical points of this function w.r.t.\ $Q$ verify $\sum_{i=1}^N \Omega_i = 0$, hence we need not enforce the last constraint: all candidate solutions are horizontal vectors. Summing up, we find that the squared geodesic distance on $\p_\emptyset$ obeys:
\begin{align}
	\dist^2([\bfR], [\hat \bfR]) & = \min_{Q\in\SOn} \sum_{i=1}^N \smallfrobnorm{\log(R_i\transpose \hat R_i^{}Q)}^2.
	\label{eq:distanchorfree}
\end{align}
Since $\SOn$ is compact, this is a well-defined quantity. Let $Q\in\SOn$ be one of the global minimizers. Then, an acceptable value for the logarithmic map is
\begin{align}
	\D\pi(\bfR)|_{\HH_{\bfR}}^{-1}[\Log_{[\bfR]}([\hat \bfR])] & = \left( R_1\log(R_1\transpose \hat R_1^{}Q), \ldots, R_N\log(R_N\transpose \hat R_N^{}Q) \right).
	\label{eq:Logpempty}
\end{align}
Under reasonable proximity conditions on $[\bfR]$ and $[\hat \bfR]$, the global maximizer $Q$ is uniquely defined, and hence so is the logarithmic map. \add{An optimal $Q$ is a \emph{Karcher mean}---or \emph{intrinsic mean} or \emph{Riemannian center of mass}---of the rotation matrices $\hat R_1\transpose R_1, \ldots, \hat R_N\transpose R_N$. Hartley et al.~{\citep{hartley2013rotation}}, among others, give a thorough overview of algorithms to compute such means as well as uniqueness conditions.}

%% file: measuresson.tex
\section{Measures, integrals and distributions on $\SOn$}
\label{sec:measuresson}

\add{To define a noise model for the synchronization measurements~{\eqref{eq:Hij}}, we now cover a notion of probability density function (pdf) over $\SOn$ and give a few examples of useful pdf's.}

Being a compact Lie group, $\SOn$ admits a unique bi-invariant Haar measure $\mu$ such that $\mu(\SOn) = 1$~\citep[Thm\,3.6, p.\,247]{boothby1986introduction}. Such a measure verifies, for all measurable subsets $S\subset\SOn$ and for all $L,R\in\SOn$, that $\mu(LSR) = \mu(S)$, where $LSR \triangleq \{ LQR : Q\in S \} \subset \SOn$. That is, the measure of a portion of $\SOn$ is invariant under left and right actions of $\SOn$. We will need something slightly more general:
\begin{lemma}[extended bi-invariance]\label{lem:biinvariance}
$\forall L, R\in\On$ such that $\det(LR) = 1$, $\forall S \subset \SOn$ measurable, $\mu(LSR) = \mu(S)$ holds.
\end{lemma}
\begin{proof}
$LSR$ is still a measurable subset of $\SOn$. Let $\mu '$ denote the Haar measure on $\On \supset \SOn$. The restriction of $\mu '$ to the measurables of $\SOn$ is still a Haar measure. By the uniqueness of the Haar measure up to multiplicative constant, there exists $\alpha > 0$ such that for all measurable subsets $T\subset \SOn$, we have $\mu(T) = \alpha\mu'(T)$. Then, $\mu(LSR) = \alpha \mu'(LSR) = \alpha\mu'(S) = \mu(S)$.
\end{proof}

From the general theory of Lebesgue integration, we get a notion of integrals over $\SOn$ associated to the measure $\mu$. Lemma~\ref{lem:biinvariance} then translates into the following property, with $f : \SOn\to\mathbb{R}$ an integrable function:
\begin{align}
	\forall L,R\in\On \textrm{ s.t.\ } \det(LR) = 1, \quad \int_\SOn f(LZR)\, \dmu(Z) & = \int_\SOn f(Z)\, \dmu(Z). 
	\label{eq:biinvarianceintegrals}
\end{align}
This property will play an important role in the sequel.

\SHORTEN{In theory, to compute integrals over $\SOn$, one could parameterize (a measure-1 subset of) the group using $d = \dim\SOn$ coordinates $x_1, \ldots, x_d$, and express $\dmu$ as $m(x_1, \ldots, x_d)\mathrm{d}x_1\mathrm{d}x_2\cdots\mathrm{d}x_d$. For $n=2$, the group $\SOn$ is identifiable with the unit circle. Parameterizing it with angles in radians, the measure is simply $\frac{1}{2\pi}\mathrm{d}x_1$. For $n=3$, the dimension of $\SOn$ is 3, and choosing an appropriate parametrization is already less obvious. Among others, Chiuso et al.\ use an explicit and somewhat involved expression for $\dmu$ in this case~\citep[eq.~(20)]{chiuso2008wide}. For larger values of $n$, the parameterization approach is doomed to become less practical.

Fortunately, in some cases, integration over $\SOn$ can be simplified.} Functions $f$ that are invariant under conjugation, meaning that for all $Z,Q\in\SOn$ we have $f(Z) = f(QZQ^{-1})$, are \emph{class functions}. When the integrand in $\int_\SOn f(Z)\, \dmu(Z)$ is a class function, we are in a position to use the Weyl integration formula specialized to $\SOn$~\citep[Exercise~18.1--2]{liegroupsbump}. All spectral functions (Assumption~\ref{assu:invariant}) are class functions (the converse is also true for $\mathrm{SO}(2k+1)$ but not for $\mathrm{SO}(2k)$). Weyl's formula comes in two flavors depending on the parity of $n$, and essentially reduces integrals on $\SOn$ to integrals over tori of dimension $\floor{n/2}$. In particular, for $n=2$ or $3$, Weyl's formula for class functions $f$ reads:
\begin{align}
	\int_{\mathrm{SO}(2)} f(Z)\, \dmu(Z) & = \frac{1}{2\pi} \int_{-\pi}^\pi f\begin{pmatrix}\cos\theta & -\sin\theta \\ \sin\theta & \cos\theta\end{pmatrix} \, \mathrm{d}\theta, \nonumber \\
	\int_{\mathrm{SO}(3)} f(Z)\, \dmu(Z) & = \frac{1}{2\pi} \int_{-\pi}^\pi f\begin{pmatrix}\cos\theta & -\sin\theta & 0 \\ \sin\theta & \cos\theta & 0 \\ 0 & 0 & 1\end{pmatrix} \, (1-\cos\theta) \, \mathrm{d}\theta.
	\label{eq:weyl23}
\end{align}
For $n=4$, Weyl's formula is a double integral:
\begin{multline}
	\int_{\mathrm{SO}(4)} f(Z)\, \dmu(Z) = \frac{1}{4(2\pi)^2} \int_{-\pi}^{\pi}\! \int_{-\pi}^{\pi} f\left(\diag\left(\begin{pmatrix}\cos\theta_1 & -\sin\theta_1 \\ \sin\theta_1 & \cos\theta_1 \end{pmatrix},\begin{pmatrix}\cos\theta_2 & -\sin\theta_2 \\ \sin\theta_2 & \cos\theta_2 \end{pmatrix}\right)\right) \\ \times |e^{i\theta_1}-e^{i\theta_2}|^2 \cdot |e^{i\theta_1}-e^{-i\theta_2}|^2 \, \mathrm{d} \theta_1 \mathrm{d} \theta_2.
	\label{eq:weyl4}
\end{multline}

Once $\SOn$ is equipped with a measure $\mu$ and accompanying integral notion, we can define distributions of random variables on $\SOn$ via probability density functions (pdf's). In general, a pdf on $\SOn$ is a nonnegative measurable function $f$ on $\SOn$ such that $\int_\SOn f(Z)\, \dmu(Z) = 1$. In this work, for convenience, we further assume pdf's are smooth and positive (Assumption~\ref{assu:smoothpositive}), as we will need to compute the derivatives of their logarithm.

\begin{example}[uniform]
The pdf associated to the \emph{uniform distribution} is $f(Z) \equiv 1$, since we normalized the Haar measure such that $\mu(\SOn) = 1$. We write $Z \sim \uniform(\SOn)$ to mean that $Z$ is a uniformly distributed random rotation. A number of algorithms exist to generate pseudo-random rotation matrices from the uniform distribution~\citep[\S{2.5.1}]{chikuse2003statistics}~\citep{diaconis1987subgroup}. Possibly one of the easiest methods to implement is the following $\mathcal{O}(n^3)$ algorithm, adapted from~\citep[Method~A, p.\,22]{diaconis1987subgroup} with implementation details as in~\citep{mezzadri2007generate} (for large $n$, see the former paper for algorithms with better complexity):
\begin{enumerate}
\item Generate $A\in\Rnn$, such that the entries $A_{ij}\sim\mathcal{N}(0,1)$ are i.i.d.\ normal random variables;
\item Obtain a QR decomposition of $A$: $QR=A$;
\item Set $Q := Q\diag(\sign(\diag(R)))$ (this ensures the mapping $A\mapsto Q$ is well-defined; see~\cite{mezzadri2007generate});
\item $Q$ is now uniform on $\On$. If $\det(Q)=-1$, permute columns 1 and 2 of $Q$. Return $Q$.
\end{enumerate}
\end{example}

\begin{example}[isotropic Langevin]\label{ex:langevindist}
The \emph{isotropic Langevin distribution} on $\SOn$ with mean $Q\in\SOn$ and concentration $\kappa\geq 0$ has pdf
\begin{align}
	f(Z) & = \frac{1}{c_n(\kappa)}\exp(\kappa\,\trace(Q\transpose Z)),
	\label{eq:langevin}
\end{align}
where $c_n(\kappa)$ is a normalization constant such that $f$ has unit mass. We write $Z \sim \langevin(Q, \kappa)$ to mean that $Z$ is a random variable with pdf~\eqref{eq:langevin}. For $\kappa=0$, $Z\sim\uniform(\SOn)$; in the limit $\kappa\to\infty$, $Z = Q$ w.p.\ 1. The isotropic Langevin distribution has a Gaussian-like shape. The Langevin pdf centered around $Q = I$ is a spectral function, i.e., it fulfills Assumption~\ref{assu:invariant}.

The larger the concentration parameter $\kappa$, the more the distribution is concentrated around the mean. By bi-invariance of $\mu$, $c_n(\kappa)$ is independent of $Q$:
\begin{align}
	c_n(\kappa) & = \int_\SOn \exp(\kappa\,\trace(Q\transpose Z))\, \dmu(Z) = \int_\SOn \exp(\kappa\,\trace(Z))\, \dmu(Z).
	\label{eq:cn}
\end{align}
Since the integrand is a class function, Weyl's integration formulas apply for any value of $n$. Using~\eqref{eq:weyl23} and~\eqref{eq:weyl4}, we work out explicit formulas for $c_n(\kappa)$, $n=2,3,4$:
\begin{align}
	c_2(\kappa) & = I_0(2\kappa), \label{eq:c2} \\
	c_3(\kappa) & = \exp(\kappa)(I_0(2\kappa)-I_1(2\kappa)), \label{eq:c3} \\
	c_4(\kappa) & = I_0(2\kappa)^2-2I_1(2\kappa)^2 + I_0(2\kappa)I_2(2\kappa), \label{eq:c4}
\end{align}
in terms of the modified Bessel functions of the first kind, $I_\nu$~\citep{wolframbessel}. See Appendix~\ref{apdx:coefficients} for details.

For $n=2$, the Langevin distribution is also known as the von Mises or Fisher distribution on the circle~\citep{mardia2000directional}. The Langevin distribution on $\SOn$ also exists in anisotropic form~\citep{chiuso2008wide}. Unfortunately, the associated pdf is no longer a spectral function even for $Q=I$, which is an instrumental property in the present work. Consequently, we do not treat anisotropic distributions. Chikuse gives an in-depth treatment of statistics on the Grassmann and Stiefel manifolds~\citep{chikuse2003statistics}, including a study of Langevin distributions on $\SOn$ as a special case.

Based on a uniform sampling algorithm on $\SOn$, it is easy to devise an acceptance-rejection scheme to sample from the Langevin distribution~\citep[\S{2.5.2}]{chikuse2003statistics}. Not surprisingly, for large values of $\kappa$, this tends to be very inefficient. Chiuso et al.\ report using a Metropolis-Hastings--type algorithm instead~\citep[\S{7}]{chiuso2008wide}. Hoff describes an efficient Gibbs sampling method to sample from a more general family of distributions on the Stiefel manifold, which can be modified to work on $\SOn$~\citep{hoff2009simulation}.

\remove{For $\langevin(I,\kappa)$, the following formula will prove useful:
\protect\begin{align}
	\log f(Z) & = \kappa\,\trace(Z) - \log c_n(\kappa), \\
	\nabla\log f(Z) & = \kappa I_n, \\
	\grad\log f(Z) & = \PP_Z(\nabla\log f(Z)) = -\kappa\, Z\skeww{Z}, \label{eq:gradlogflangevin}
\end{align}
where $\nabla$ denotes the usual gradient in the embedding space $\Rnn$, $\grad$ denotes the gradient on $\SOn$ and $\PP_Z$~\eqref{eq:projSOn} is the orthogonal projector from the embedding space onto the tangent space  $\T_Z\SOn$, with $\skeww{Z} \triangleq (Z-Z\transpose)/2$ the skew-symmetric part of $Z$---see Section~{\ref{sec:geometry}}.}
\end{example}

The set of pdf's is closed under convex combinations, as is the set of functions satisfying our assumptions~\ref{assu:smoothpositive} and~\ref{assu:invariant}. We could therefore consider mixtures of Langevin distributions around $I$ with various concentrations. In the next example, we combine the Langevin distribution and the uniform distribution. To the best of our knowledge, this is the first description of a heavy-tail--type distribution on $\SOn$. Such a distribution may prove useful for any application involving outliers in rotation estimation.
\begin{example}[isotropic Langevin with outliers] \label{ex:langevinoutliersdist}
We define the \emph{isotropic Langevin distribution with outliers} on $\SOn$ with mean $Q\in\SOn$, concentration $\kappa>0$ and outlier probability $1-p\in [0, 1]$ via the pdf
\begin{align}
	f(Z) & = \frac{p}{c_n(\kappa)}\exp(\kappa\,\trace(Q\transpose Z)) + (1-p).
	\label{eq:langevinoutliers}
\end{align}
For $Z$ distributed as such, we write $Z\sim\langout(Q, \kappa, p)$. For $p=1$, this is the isotropic Langevin distribution. For $p=0$, this is the uniform distribution. Notice that $f$ is a spectral function for $Q=I$.

A random rotation sampled from $\langout(Q, \kappa, p)$ is, with probability $p$, sampled from $\langevin(Q, \kappa)$, and with probability $1-p$, sampled uniformly at random. Measurements of $Q$ distributed in such a way are outliers with probability $1-p$, i.e., bear no information about $Q$ w.p.\ $1-p$. For $0<p<1$, the $\langout$ is a kind of heavy-tailed distribution on $\SOn$.

\remove{For $\langout(I,\kappa,p)$, the following will prove useful:
\protect\begin{align}
	\grad\log f(Z) & = -\frac{p\kappa\,\exp(\kappa\,\trace\, Z)}{c_n(\kappa)f(Z)} \, Z \skeww{Z},
	\label{eq:gradlogflangout}
\end{align}
with $\skeww{Z} = (Z-Z\transpose)/2$ denoting the skew-symmetric part of $Z$.}
\end{example}

To conclude this section, we remark more broadly that all isotropic distributions around the identity matrix have a spectral pdf. Indeed, let $f : \SOn \to \mathbb{R}$ be isotropic w.r.t.\ the geodesic distance on $\SOn$, $\dist(R_1, R_2) = \frobnorm{\log(R_1\transpose R_2^{})}$~\eqref{eq:distSOn}. That is, there is a function $\tilde f$ such that $f(Z) = \tilde f(\dist(I, Z)) = \tilde f(\frobnorm{\log Z})$. It is then obvious that $f(QZQ\transpose) = f(Z)$ for all $Q\in\On$ since $\log(QZQ\transpose) = Q\log(Z)Q\transpose$. The same holds for the embedded distance $\dist(R_1, R_2) = \frobnorm{R_1-R_2}$. This shows that the assumptions proposed in Section~\ref{sec:synchro} include many interesting distributions.

\SHORTEN{The Langevin-based distributions here described are isotropic w.r.t.\ the embedded distance. For $n=3$, the notions of isotropy w.r.t.\ both distances coincide. In fact, for $n=3$, the set of isotropic pdf's and the set of spectral pdf's are the same, since all of these functions can only depend on the amplitude $0\leq \theta\leq \pi$ of their input rotation. For $n=2$, owing to the commutativity of rotations in the plane, all pdf's are class functions, and all pdf's verifying $f(Z) = f(Z\transpose)$ are spectral functions.}

Similarly we establish that all spectral pdf's have zero bias around the identity matrix $I$. The bias is the tangent vector (skew-symmetric matrix) $\Omega = \expectt{\Log_I(Z)}$, with $Z \sim f$, $f$ spectral. Since $\Log_I(Z) = \log(Z)$~\eqref{eq:ExpLogSOn}, we find, with a change of variable $Z := QZQ\transpose$ going from the first to the second integral, that for all $Q\in\On$:
\begin{align}
	\Omega & = \int_{\SOn} \log(Z)\,f(Z)\dmu(Z) = \int_{\SOn} \log(QZQ\transpose)\,f(Z)\dmu(Z) = Q\Omega Q\transpose.
\end{align}
\remove{Indeed, $\log(QZQ\transpose) = Q\log(Z)Q\transpose$. }Since skew-symmetric matrices are normal matrices and since $\Omega$ and $\Omega\transpose = -\Omega$ have the same eigenvalues, we may choose $Q\in\On$ such that $Q\Omega Q\transpose = -\Omega$.
Therefore, $\Omega = -\Omega = 0$. As a consequence, it is only possible to treat unbiased measurements under the assumptions we make in this paper.

%% file: fimsynch.tex
\section{The Fisher information matrix for synchronization}
\label{sec:fimsynch}

As described in Section~\ref{sec:synchro}, the relative rotation measurements $H_{ij}^{} = Z_{ij}^{}R_i^{}R_j\transpose$~\eqref{eq:Hij} reveal information about the true (but unknown) rotations $R_1, \ldots, R_N$. The Fisher information matrix (FIM) we compute here encodes how much information these measurements contain on average. In other words, the FIM is an assessment of the quality of the measurements we have at our disposal for the purpose of estimating the sought parameters. The FIM will be instrumental in deriving Cramér-Rao bounds in the next section.

The FIM is a standard object of study for estimation problems on Euclidean spaces. In the setting of synchronization of rotations, the parameter space is a manifold and we thus need a more general definition of the FIM. We quote, mutatis mutandis, the definition of FIM as stated in~\citep{crbsubquot} following~\citep{smith2005covariance}:
\begin{definition}[FIM]\label{def:FIM}
Let $\p$ be the parameter space of an estimation problem and let $\theta\in\p$ be the (unknown) parameter. Let $f(y;\theta)$ be the pdf of the measurement $y$ conditioned by $\theta$. The log-likelihood function $L : \p \to \mathbb{R}$ is $L(\theta) = \log f(y; \theta)$. Let $e = \{e_1, \ldots, e_d\}$ be an orthonormal basis of the tangent space $\T_\theta\p$ w.r.t.\ the Riemannian metric $\inner{\cdot}{\cdot}_\theta$. The Fisher information matrix of the estimation problem on $\p$ w.r.t.\ the basis $e$ is defined by:
\begin{align}
	F_{ij} & = \expectt{\inner{\grad\,L(\theta)}{e_i}_\theta \cdot \inner{\grad\,L(\theta)}{e_j}_\theta}.
	\label{eq:FIM}
\end{align}
Expectations are taken w.r.t.\ the measurement $y$.
\end{definition}

\add{We will now compute the FIM for the synchronization problem. Much of the technicalities involved originate in the non-commutativity of rotations. It is helpful and informative to first go through this section with the special case $\mathrm{SO}(2)$ in mind. Doing so, rotations commute and the space of rotations has dimension $d = 1$, so that one can reach the final result more directly.}

Definition~\ref{def:FIM} stresses the role of the gradient of the log-likelihood function $L$~\eqref{eq:L}, $\grad\,L(\hat \bfR)$, a tangent vector in $\T_{\hat \bfR}\p$. The $i^\textrm{th}$ component of this gradient, that is, the gradient of the mapping $\hat R_i \mapsto L(\hat \bfR)$ with $\hat R_{j\neq i}$ fixed, is a vector field on $\SOn$ which can be written as:
\begin{align}
	\grad_i\,L(\hat \bfR) 
							 & = \sum_{j\in V_i} \left[ \grad\log f_{ij}(H_{ij}^{}\hat R_j^{}\hat R_i\transpose) \right]\transpose H_{ij} \hat R_j.
\label{eq:gradiLhatR}
\end{align}
Evaluated at the true rotations $\bfR$, this component becomes
\begin{align}
	\grad_i\,L(\bfR) 	& = \sum_{j\in V_i} \left[ \grad\log f_{ij}(Z_{ij}) \right]\transpose Z_{ij} R_i.
\end{align}
The vector field $\grad\log f_{ij}$ on $\SOn$ may be factored into:
\begin{align}
	\grad\log f_{ij}(Z) & = Z G_{ij}\transpose\,(Z),
	\label{eq:defG}
\end{align}
where $G_{ij} : \SOn \mapsto \son$ is a mapping that will play an important role in the sequel. In particular, the $i^\textrm{th}$ gradient component now takes the short form:
\begin{align}
	\grad_i\,L(\bfR) 	& = \sum_{j\in V_i} G_{ij}(Z_{ij}) R_i.
\end{align}

Let us consider a canonical orthonormal basis of $\son$: $(E_1, \ldots, E_d)$, with $d = n(n-1)/2$. For $n = 3$, we pick this one:
\begin{align}
	E_1 & = \frac{1}{\sqrt{2}}\begin{pmatrix} 0 & 1 & 0 \\ -1 & 0 & 0 \\ 0 & 0 & 0 \end{pmatrix}, & E_2 & = \frac{1}{\sqrt{2}}\begin{pmatrix} 0 & 0 & -1 \\ 0 & 0 & 0 \\ 1 & 0 & 0 \end{pmatrix}, & E_3 & = \frac{1}{\sqrt{2}}\begin{pmatrix} 0 & 0 & 0 \\ 0 & 0 & 1 \\ 0 & -1 & 0 \end{pmatrix}.
\label{eq:basisSOt}
\end{align}
An obvious generalization yields similar bases for other values of $n$. We can transport this canonical basis into an orthonormal basis for the tangent space $\T_{R_i}\SOn$ as $(R_iE_1, \ldots, R_iE_d)$.
\remove{To compute the Fisher information matrix, we need an orthonormal basis for the tangent space $\T_{\bfR}\p$ to the manifold $\p = \SOn \times \cdots \times \SOn$.} Let us also fix an orthonormal basis for the tangent space at $\bfR = (R_1, \ldots, R_N)$ of $\p$, as
\begin{multline}
	(\bfxi_{ik})_{i=1\ldots N, k=1\ldots d}, \textrm{ with } \bfxi_{ik} = (0, \ldots, 0, R_iE_k, 0, \ldots, 0), \\ \textrm{ a zero vector except for the $i^\textrm{th}$ component equal to } R_iE_k.
	\label{eq:basisTRP}
\end{multline}
The FIM w.r.t.\ this basis is composed of $N\times N$ blocks of size $d\times d$. Let us index the $(k, \ell)$ entry inside the $(i, j$) block as $F_{ij,k\ell}$. Accordingly, the matrix $F$ at $\bfR$ is defined by (see Definition~\ref{def:FIM}):
\begin{align}
	F_{ij,k\ell} & = \expectt{\innersmall{\grad\,L(\bfR)}{\bfxi_{ik}} \cdot \innersmall{\grad\,L(\bfR)}{\bfxi_{j\ell}}} \nonumber\\[2mm]
						 & = \expectt{\innersmall{\grad_i\,L(\bfR)}{R_iE_k} \cdot \innersmall{\grad_j\,L(\bfR)}{R_jE_\ell}} \nonumber\\[2mm]
						 & = \sum_{r\in V_i} \sum_{s\in V_j} \expectt{\innerbig{G_{ir}(Z_{ir})}{R_i^{} E_k^{} R_i\transpose} \cdot \innerbig{G_{js}(Z_{js})}{R_j^{} E_\ell^{} R_j\transpose}}. \label{eq:Fijkldifficult}
\end{align}

We prove that, in expectation, the mappings $G_{ij}$~\eqref{eq:defG} are zero. This fact is directly related to the standard result from estimation theory stating that the average \emph{score} $V(\theta) = \expectt{\grad\log f(y;\theta)}$ for a given parameterized probability density function $f$ is zero.
\begin{lemma} \label{lem:avgscore}
Given a smooth probability density function $f : \SOn \to \mathbb{R}^+$ and the mapping $G : \SOn \to \son$ such that $\grad\log f(Z) = ZG(Z)$, we have that $\expectt{G(Z)} = 0$, where expectation is taken w.r.t.\ $Z$, distributed according to $f$.
\end{lemma}
\begin{proof}
Define $h(Q) = \int_{\SOn} f(ZQ) \, \mathrm{d}\mu(Z)$ for $Q\in\SOn$. Since $f$ is a probability density function, bi-invariance of $\mu$~\eqref{eq:biinvarianceintegrals} yields $h(Q) \equiv 1$. Taking gradients with respect to the parameter $Q$, we get:
$$
	0 = \grad\,h(Q) = \int_{\SOn} \grad_Q f(ZQ)\ \dmu(Z) = \int_{\SOn} Z\transpose \grad f(ZQ)\ \dmu(Z).
$$
With a change of variable $Z := ZQ$, by bi-invariance of $\mu$, we further obtain:
$$
	\int_{\SOn} Z\transpose \grad f(Z)\ \dmu(Z) = 0.
$$
Using this last result and the fact that $\grad\log f(Z) = \frac{1}{f(Z)}\grad\,f(Z)$, we conclude:
\begin{align*}
	\expectt{G(Z)} = \int_{\SOn} Z\transpose \grad\log f(Z)\ f(Z) \dmu(Z) = \int_{\SOn} Z\transpose \grad f(Z)\ \dmu(Z) = 0. \mbox{\qedhere}
\end{align*}
\end{proof}

We now invoke Assumption~\ref{assu:independence} (independence). Independence of $Z_{ij}$ and $Z_{pq}$ for two distinct edges $(i,j)$ and $(p,q)$ implies that, for any two functions $\phi_1, \phi_2 : \SOn \to \mathbb{R}$, it holds that $$\expectt{\phi_1(Z_{ij}) \phi_2(Z_{pq})} = \expectt{\phi_1(Z_{ij})}\expectt{\phi_2(Z_{pq})},$$ provided all involved expectations exist. Using both this and Lemma~\ref{lem:avgscore}, most terms in \eqref{eq:Fijkldifficult} vanish and we obtain a simplified expression for the matrix $F$:
\begin{align}
	F_{ij,k\ell} & =
\begin{cases}
	\displaystyle{\sum_{r\in V_i} \expectt{\inner{G_{ir}(Z_{ir})}{R_i^{} E_k^{} R_i\transpose} \cdot \inner{G_{ir}(Z_{ir})}{R_i^{} E_\ell^{} R_i\transpose}}}, & \textrm{ if } i=j, \\[5mm]
	\displaystyle{\expectt{\inner{G_{ij}(Z_{ij})}{R_i^{} E_k^{} R_i\transpose} \cdot \inner{G_{ji}(Z_{ji})}{R_j^{} E_\ell^{} R_j\transpose}}}, & \textrm{ if } i\neq j \textrm{ and } (i,j) \in \mathcal{E}, \\[5mm]
	0, & \textrm{ if } i\neq j \textrm{ and } (i,j) \notin \mathcal{E}.
\end{cases}
\label{eq:Fijklsimpler}
\end{align}
We further manipulate the second case, which involves both $G_{ij}$ and $G_{ji}$, by noting that those are deterministically linked. Indeed, by symmetry of the measurements ($H_{ij}^{} = H_{ji}\transpose$), we have that (i) $Z_{ji} = R_j^{}R_i\transpose Z_{ij}\transpose R_i^{} R_j\transpose$ and (ii) $f_{ij}(Z_{ij}) = f_{ji}(Z_{ji})$. Invoking Assumption~\ref{assu:invariant}, since $Z_{ij}$ and $Z_{ji}$ have the same eigenvalues, it follows that $f_{ij}(Z) = f_{ji}(Z)$ for all $Z\in\SOn$. As a by-product, it also holds that $G_{ij}(Z) = G_{ji}(Z)$ for all $Z\in\SOn$. Still under Assumption~\ref{assu:invariant}, we show in Appendix~\ref{apdx:GijGji} that
\begin{align}
	\forall Q\in\On, \quad G_{ij}(QZQ\transpose) = Q\,G_{ij}(Z)\,Q\transpose, \quad \textrm{ and } \quad G_{ij}(Z\transpose) = -G_{ij}(Z).
	\label{eq:Gclass}
\end{align}
Combining these observations, we obtain:
\begin{align}
	G_{ji}(Z_{ji}) = G_{ij}(Z_{ji}) = G_{ij}(R_j^{}R_i\transpose \, Z_{ij}\transpose \, R_i^{} R_j\transpose) = -R_j^{}R_i\transpose \, G_{ij}( Z_{ij} ) \, R_i^{} R_j\transpose.
\end{align}
The minus sign, which plays an important role in the structure of the FIM, comes about via the skew-symmetry of $G_{ij}$. The following identity thus holds:
\begin{align}
	\inner{G_{ji}(Z_{ji})}{R_j^{} E_\ell^{} R_j\transpose} & = -\inner{G_{ij}^{}(Z_{ij}^{})}{R_i^{} E_\ell^{} R_i\transpose}.
\label{eq:GijGji2}
\end{align}
This can advantageously be plugged into~\eqref{eq:Fijklsimpler}.

We set out to describe the expectations appearing in~\eqref{eq:Fijklsimpler}, which will take us through a couple of lemmas. Let us, for a certain pair $(i,j)\in\mathcal{E}$, introduce the functions $h_k : \SOn \to \mathbb{R}$, $k=1\ldots d$:
\begin{align}
	h_k(Z) \triangleq \inner{G_{ij}(Z)}{R_i^{}E_k^{}R_i\transpose},
	\label{eq:hk}
\end{align}
where we chose to not overload the notation $h_k$ with an explicit reference to the pair $(i,j)$, as this will always be clear from the context. We may rewrite the FIM in terms of the functions $h_k$, starting from~\eqref{eq:Fijklsimpler} and incorporating~\eqref{eq:GijGji2}:
\begin{align}
	F_{ij,k\ell} & =
\begin{cases}
	\displaystyle{\sum_{r\in V_i} \expectt{h_k(Z_{ir}) \cdot h_\ell(Z_{ir})}}, & \textrm{ if } i=j, \\[5mm]
	\displaystyle{-\expectt{h_k(Z_{ij}) \cdot h_\ell(Z_{ij})}}, & \textrm{ if } i\neq j \textrm{ and } (i,j) \in \mathcal{E}, \\[5mm]
	0, & \textrm{ if } i\neq j \textrm{ and } (i,j) \notin \mathcal{E}.
\end{cases}
\label{eq:Fijklhk}
\end{align}
Another consequence of Assumption~\ref{assu:invariant} is that the functions $h_k(Z)$ and $h_\ell(Z)$ are uncorrelated for $k\neq \ell$, where $Z$ is distributed according to the density $f_{ij}$. As a consequence, $F_{ij,k\ell} = 0$ for $k\neq \ell$, i.e., the $d\times d$ blocks of $F$ are diagonal. We establish this fact in Lemma~\ref{lem:decorrelation}, right after a technical lemma.
\begin{lemma} \label{lem:Pkl}
Let $E, E' \in \son$ such that $E_{ij} = -E_{ji} = 1$ and $E'_{k\ell} = -E'_{\ell k} = 1$ (all other entries are zero), with $\inner{E}{E'} = 0$, i.e., $\{i,j\} \neq \{k,\ell\}$. Then, there exists $P\in\On$ a signed permutation such that $P\transpose E P = E'$ and $P\transpose E' P = -E$.
\end{lemma}
\begin{proof}
See Appendix~\ref{apdx:Pkl}.\qedhere
\end{proof}
\begin{lemma} \label{lem:decorrelation}
Let $Z\in\SOn$ be a random variable distributed according to $f_{ij}$. The random variables $h_k(Z)$ and $h_\ell(Z)$, $k\neq \ell$, as defined in~\eqref{eq:hk} have zero mean and are uncorrelated, i.e., $\expectt{h_k(Z)} = \expectt{h_\ell(Z)} = 0$ and $\expectt{h_k(Z)\cdot h_\ell(Z)} = 0$. Furthermore, it holds that $\expectt{h_k^2(Z)} = \expectt{h_\ell^2(Z)}$.
\end{lemma}
\begin{proof}
The first part follows directly from Lemma~\ref{lem:avgscore}. We show the second part. Consider a signed permutation matrix $P_{k\ell}\in\On$ such that $P_{k\ell}\transpose\, E_k^{} P_{k\ell}^{} = E_\ell^{}$ and $P_{k\ell}\transpose\, E_\ell^{} P_{k\ell}^{} = -E_k^{}$. Such a matrix always exists according to Lemma~\ref{lem:Pkl}. Then, identity~\eqref{eq:Gclass} yields:
\begin{align}
	h_k(R_i^{}P_{k\ell}^{}R_i\transpose \  Z \  R_i^{}P_{k\ell}\transpose\,R_i\transpose) & = \inner{G_{ij}(Z)}{R_i^{} P_{k\ell}\transpose\, E_k^{} P_{k\ell}^{} R_i\transpose} = h_\ell(Z).
\end{align}
Likewise,
\begin{align}
	h_\ell(R_i^{}P_{k\ell}^{}R_i\transpose \ Z \ R_i^{}P_{k\ell}\transpose\,R_i\transpose) & = -h_k(Z).
\end{align}
These identities as well as the (extended) bi-invariance~\eqref{eq:biinvarianceintegrals} of the Haar measure $\mu$ on $\SOn$ and the fact that $f_{ij}$ is a spectral function yield, using the change of variable $Z := R_i^{}P_{k\ell}^{}R_i\transpose \  Z \  R_i^{}P_{k\ell}\transpose\,R_i\transpose$ going from the first to the second integral:
\begin{align}
	\expectt{h_k(Z)\cdot h_\ell(Z)} & = \int_{\SOn} h_k(Z) h_\ell(Z) \ f_{ij}(Z)\dmu(Z) \\
	& =  \int_{\SOn} -h_\ell(Z) h_k(Z) \ f_{ij}(Z)\dmu(Z) = -\expectt{h_k(Z)\cdot h_\ell(Z)}.
\end{align}
Hence, $\expectt{h_k(Z)\cdot h_\ell(Z)} = 0$. We prove the last statement using the same change of variable:
\begin{align*}
	\expectt{h_k^2(Z)} & = \int_{\SOn} h_k^2(Z) \ f_{ij}(Z)\dmu(Z) = \int_{\SOn} h_\ell^2(Z) \ f_{ij}(Z)\dmu(Z) = \expectt{h_\ell^2(Z)}.
\end{align*}
We note that, more generally, it can be shown that the $h_k$'s are identically distributed.
\end{proof}

\SHORTEN{A stronger result states that $h_k$ and $h_\ell$ are identically distributed. A proof is given for completeness, although we will not use this result here:
\begin{lemma} \label{lem:identicaldistribs}
(Sequel of Lemma~\ref{lem:decorrelation}) $h_k(Z)$ and $h_\ell(Z)$ are two identically distributed random variables.
\end{lemma}
\begin{proof}
Let $\mathcal{E}_k(s) = \{ Z\in\SOn : h_k(Z) \leq s \}$. Consider still the same change of variable as in the proof of Lemma~\ref{lem:decorrelation}, $Z = R_i^{}P_{k\ell}^{}R_i\transpose \  Z' \  R_i^{}P_{k\ell}\transpose\,R_i\transpose$. It follows easily that
\begin{align*}
	Z\in\mathcal{E}_k(s) \Leftrightarrow h_k(Z) \leq s \Leftrightarrow h_\ell(Z') \leq s \Leftrightarrow Z'\in\mathcal{E}_\ell(s).
\end{align*}
Hence, this change of variable is a diffeomorphism between $\mathcal{E}_k(s)$ and $\mathcal{E}_\ell(s)$. Using again that $f_{ij}$ is a spectral function and the bi-invariance of $\mu$, we obtain that:
\begin{align*}
	P[h_k(Z)\leq s] & = \int_{\mathcal{E}_k(s)} f_{ij}(Z)\,\dmu(Z) 
						  = \int_{\mathcal{E}_\ell(s)} f_{ij}(Z')\,\dmu(Z') = P[h_\ell(Z)\leq s],
\end{align*}
where $P[E]$ is the probability of event $E$ occurring.
\end{proof}}

The skew-symmetric matrices $(R_i^{}E_1^{}R_i\transpose, \ldots, R_i^{}E_d^{}R_i\transpose)$ form an orthonormal basis of the Lie algebra $\son$. Consequently, we may expand each mapping $G_{ij}$ in this basis and express its squared norm as:
\begin{align}
	G_{ij}(Z) & = \sum_{k=1}^d h_k(Z) \cdot R_i^{}E_k^{}R_i\transpose, & \|G_{ij}(Z)\|^2 & = \sum_{k=1}^d h_k^2(Z).
\end{align}
Since by Lemma~\ref{lem:decorrelation} the quantity $\expectt{h_k^2(Z_{ij})}$ does not depend on $k$, it follows that:
\begin{align}
	\expectt{h_k^2(Z_{ij}^{})} & = \frac{1}{d} \expectt{\|G_{ij}(Z_{ij})\|^2}, \quad k = 1\ldots d.
\end{align}
This further shows that the $d\times d$ blocks that constitute the FIM have constant diagonal. Hence, $F$ can be expressed as the Kronecker product ($\otimes$) of some matrix with the identity $I_d$. Let us define the following (positive) weights on the edges of the measurement graph:
\begin{align}
	w_{ij} = w_{ji} = \expectt{\|G_{ij}(Z_{ij})\|^2} \triangleq \expectt{\|\grad\log f_{ij}(Z_{ij})\|^2},\quad \forall (i,j)\in\mathcal{E}.
	\label{eq:infoweight}
\end{align}
\add{Let $\mathcal{A}\in\mathbb{R}^{N\times N}$ be the adjacency matrix of the measurement graph with $\mathcal{A}_{ij} = w_{ij}$ if $(i,j)\in\mathcal{E}$ and $\mathcal{A}_{ij} = 0$ otherwise, and let $\mathcal{D}\in\mathbb{R}^{N\times N}$ be the diagonal degree matrix such that $\mathcal{D}_{ii} = \sum_{j\in V_i} w_{ij}$. Then, the \emph{weighted Laplacian matrix} $\mathcal{L} = \mathcal{D} - \mathcal{A}$, $\mathcal{L} = \mathcal{L}\transpose \succeq 0$, is given by:}
\begin{align}
	\mathcal{L}_{ij} = 
\begin{cases}
	\sum_{r\in V_i}w_{ir}, & \textrm{ if } i=j, \\
	-w_{ij}, & \textrm{ if } i\neq j \textrm{ and } (i,j) \in \mathcal{E}, \\
	0, & \textrm{ if } i\neq j \textrm{ and } (i,j) \notin \mathcal{E}.
\end{cases}
\label{eq:laplacian}
\end{align}
It is now apparent that the matrix $F \in \mathbb{R}^{dN\times dN}$ is tightly related to $\mathcal{L}$. We summarize this in the following theorem.

\begin{theorem}[FIM for synchronization]\label{thm:FIM}
Let $R_1, \ldots, R_N \in \SOn$ be unknown but fixed rotations and let $H_{ij}^{} = Z_{ij}^{}R_i^{}R_j\transpose$ for $(i,j)\in\mathcal{E}, i<j$, with the $Z_{ij}$'s random rotations which fulfill Assumptions \ref{assu:smoothpositive}--\ref{assu:invariant}. Consider the problem of estimating the $R_i$'s given a realization of the $H_{ij}$'s. The Fisher information matrix (Definition~\ref{def:FIM}) of that estimation problem with respect to the basis~\eqref{eq:basisTRP} is given by
\begin{align}
	F = \frac{1}{d} (\mathcal{L}\otimes I_d),
	\label{eq:Fgeneral}
\end{align}
where $\otimes$ denotes the Kronecker product, $d = \dim \SOn = n(n-1)/2$, $I_d$ is the $d\times d$ identity matrix and $\mathcal{L}$ is the weighted Laplacian matrix~\eqref{eq:laplacian} of the measurement graph.
\end{theorem}

The Laplacian matrix has a number of properties, some of which will yield nice interpretations when deriving the Cramér-Rao bounds. One remarkable fact is that the Fisher information matrix does not depend on $\bfR = (R_1, \ldots, R_N)$, the set of true rotations. This is an appreciable property seen as $\bfR$ is unknown in practice. This stems from the strong symmetries in our problem.

Another important feature of this FIM is that it is rank deficient. Indeed, for a connected measurement graph, $\mathcal{L}$ has exactly one zero eigenvalue (and more if the graph is disconnected) associated to the vector of all ones, $\mathds{1}_N$. The null space of the FIM is thus composed of all vectors of the form $\mathds{1}_N\otimes t$, with $t\in\mathbb{R}^d$ arbitrary. This corresponds to the vertical spaces of $\p$ w.r.t.\ the equivalence relation~\eqref{eq:equivrelation}, i.e., the null space consists in all tangent vectors that move all rotations $R_i$ in the same direction, leaving their relative positions unaffected. This makes perfect sense: the distribution of the measurements $H_{ij}$ is also unaffected by such changes, hence the FIM, seen as a quadratic form, takes up a zero value when applied to the corresponding vectors. We will need special tools to deal with this (structured) singularity when deriving the CRB's in the next section.

Notice how Assumption~\ref{assu:independence} (independence) gave $F$ a block structure based on the sparsity pattern of the Laplacian matrix, while Assumption~\ref{assu:invariant} (spectral pdf's) made each block proportional to the $d\times d$ identity matrix and made $F$ independent of $\bfR$.

In the two following examples, we explicitly compute the weights $w_{ij}$~\eqref{eq:infoweight} associated to two types of noise distributions: (1) unbiased isotropic Langevin distributions (akin to Gaussian noise), and (2) a mix of the former distribution and complete outliers (uniform distribution)---see Section~\ref{sec:measuresson}. Combining formulas for the weights $w_{ij}$ and equations~\eqref{eq:laplacian} and~\eqref{eq:Fgeneral}, one can compute the Fisher information matrix explicitly.

\begin{example}[Langevin distributions] \label{ex:langevin}
(Continued from Example~\ref{ex:langevindist}) Considering the isotropic Langevin distribution $f$~\eqref{eq:langevin}, $\grad\log f(Z) = -\kappa\, Z\skeww{Z}$ and we find that the weight $w$ associated to this noise distribution is a function $\alpha_n(\kappa)$ given by:
\begin{align}
	w = \alpha_n(\kappa) = \expectt{\|\grad\log f(Z)\|^2} & = \frac{\kappa^2}{4}\int_{\SOn} \|Z-Z\transpose\|^2 \ f(Z)\dmu(Z).
\end{align}
Since the integrand is again a class function, we may evaluate this integral using Weyl's integration formulas---see Section~\ref{sec:measuresson} and Appendix~\ref{apdx:coefficients} for an example. In particular, for $n=2$ and $3$, identities~\eqref{eq:weyl23} apply and we derive the following expressions:
\begin{align}
	\alpha_2(\kappa) & = \kappa\frac{I_1(2\kappa)}{I_0(2\kappa)}, & \alpha_3(\kappa) & = \frac{\kappa}{2}\, \frac{(2-\kappa)I_1(2\kappa) + \kappa I_3(2\kappa)}{I_0(2\kappa) - I_1(2\kappa)}.
\end{align}
The functions $I_\nu(z)$ are the modified Bessel functions of the first kind~\eqref{eq:bessel}. We used the formulas for the normalization constants $c_2(\kappa)$~\eqref{eq:c2} and $c_3(\kappa)$~\eqref{eq:c3} as well as the identity $I_1(2\kappa) = \kappa(I_0(2\kappa)-I_2(2\kappa))$.

For the special case $n=2$, taking the concentrations for all measurements to be equal, we find that the FIM is proportional to the unweighted Laplacian matrix $D-A$, with $D$ the degree matrix and $A$ the adjacency matrix of the measurement graph. This particular result was shown before via another method in~\citep{howard2010estimation}. For the derivation in the latter work, commutativity of rotations in the plane is instrumental, and hence the proof method does not---at least in the proposed form---transfer to $\SOn$ for $n \geq 3$.
\end{example}

\begin{example}[Langevin/outlier distributions] \label{ex:langevinoutliers}
(Continued from Example~\ref{ex:langevinoutliersdist})\remove{ Given a concentration $\kappa > 0$ and a probability $p\in[0, 1]$, the pdf
\protect\begin{align}
	f(Z) = \frac{p}{c_n(\kappa)} \exp(\kappa\,\trace\,Z) + 1-p
\end{align}
corresponds to picking a rotation, with probability $p$, around the identity matrix following an isotropic Langevin distribution with concentration $\kappa$, or, with probability $1-p$, uniformly at random on $\SOn$.} Considering the pdf $f$~\eqref{eq:langevinoutliers}, we find $\grad\log f(Z) = \frac{1}{f(Z)} \frac{p\kappa\,\exp(\kappa\trace\,Z)}{c_n(\kappa)} Z\skeww{Z}$. Thus, the information weight is given by
\begin{align}
	w = \alpha_n(\kappa, p) = \int_{\SOn} \frac{1}{f(Z)} \left(\frac{p\kappa\,\exp(\kappa\,\trace\,Z)}{2c_n(\kappa)}\right)^2 \|Z-Z\transpose\|^2 \dmu(Z).
\end{align}
Some algebra using the material in Section~\ref{sec:measuresson} yields, for $n = 2$ and $3$:
\begin{align}
	\alpha_2(\kappa, p) & = \frac{(p\kappa)^2}{c_2(\kappa)} \frac{1}{\pi} \int_0^\pi \frac{(1-\cos 2\theta)\exp(4\kappa\cos\theta)}{p\exp(2\kappa\cos\theta) + (1-p)c_2(\kappa)} \mathrm{d}\theta, \label{eq:alphatwokp} \\[3mm]
	\alpha_3(\kappa, p) & = \frac{(p\kappa)^2 \exp(2\kappa)}{c_3(\kappa)} \frac{1}{\pi} \int_0^\pi \frac{(1-\cos 2\theta)(1-\cos\theta)\exp(4\kappa\cos\theta)}{p\exp(\kappa(1+2\cos\theta)) + (1-p)c_3(\kappa)} \mathrm{d}\theta. \label{eq:alphathreekp}
\end{align}
These integrals may be evaluated numerically. The same machinery goes through for $n \geq 4$.
\end{example}

%% file: crbsynch.tex
\section{Cramér-Rao bounds for synchronization}
\label{sec:crbsynch}

Classical Cramér-Rao bounds (CRB's) give a lower bound on the covariance matrix $C$ of any unbiased estimator for an estimation problem in $\Rn$. In terms of the Fisher information matrix $F$ of that problem, the classical result reads $C \succeq F^{-1}$. In our setting, an estimation problem on a manifold with singular $F$, we need to resort to a more general statement of the CRB.

First off, because the parent parameter space $\p$ is a manifold instead of a Euclidean space, we need generalized definitions of bias and covariance of estimators---we will quote them momentarily. And because manifolds are usually curved---as opposed to Euclidean spaces which are flat---the CRB takes up the form $C \succeq F^{-1} + \textrm{ curvature terms}$~\citep{smith2005covariance}. We will compute the curvature terms and show that they become negligible at large signal-to-noise ratios (SNR).

Secondly, owing to the \change{indeterminacies in our estimation problem (based on the $H_{ij}$'s only, we can only recover the $R_i$'s up to a global rotation)}{global rotation ambiguity in synchronization}, the FIM~\eqref{eq:Fgeneral} is singular, with a kernel that is identifiable with the vertical space~\eqref{eq:vertical}. In~\citep{crbsubquot}, CRB's are provided for this scenario by looking at the estimation problem either on the submanifold $\p_A$ (where indeterminacies have been resolved by fixing anchors) or on the quotient space $\p_\emptyset$ (where each equivalence class of rotations is regarded as one point).

We should bear in mind that intrinsic Cramér-Rao bounds are fundamentally asymptotic bounds for large SNR~\citep{smith2005covariance}. At low SNR, the bounds may fail to capture features of the estimation problem that become dominant for large errors. In particular, since the parameter spaces $\p_A$ and $\p_\emptyset$ are compact, there is an upperbound on how badly one can estimate the true rotations. Because of their local nature (intrinsic CRB's result from a small-error analysis), the bounds we establish here are unable to capture this essential feature.\remove{ and should not be used in low SNR regimes.}

\add{In the sequel, the proviso \emph{at large SNR} thus designates noise levels such that efficient estimators commit errors small enough that the intrinsic CRB analysis holds. For reasons that will become clear in this section, for anchor-free synchronization, we define a notion of SNR as the quantity
\protect\begin{align}
	\mathrm{SNR}_\emptyset & =
\frac{(N-1)\expectt{\dist^2(Z_{\mathrm{uni}}, I_n)}}{d^2\trace(\mathcal{L}^\dagger)},
	\label{eq:SNRanchorfree}
\end{align}
where the expectation is taken w.r.t.\ $Z_{\mathrm{uni}}$, uniformly distributed over $\SOn$. The numerator is a baseline which corresponds to the variance of a random estimator---see Section~{\ref{subsec:asymptoticbound}}. The denominator has units of variance as well and is small when the measurement graph is well connected by good measurements. An SNR can be considered large if $\mathrm{SNR}_\emptyset \gg 1$. For anchored synchronization, a similar definition holds with $\mathcal{L}$ replaced by the masked Laplacian $\mathcal{L}_A$~{\eqref{eq:LA}} and $N-1$ replaced by $N-|A|$.
}

We now give a definition of unbiased estimator. After this, we differentiate between the anchored and the anchor-free scenarios to establish CRB's.
\begin{definition}[unbiased estimator]\label{def:unbiasedestim}
Let $\p$ (a Riemannian manifold) be the parameter space of an estimation problem\add{ and let $\M$ be the probability space to which measurements belong}. Let $f(y;\theta)$ be the pdf of the measurement $y\in\M$ conditioned by $\theta\in\p$. An \emph{estimator} is a mapping $\hat\theta : \M \to \p$ assigning a parameter $\hat\theta(y)$ to every possible realization of the measurement $y$. The \emph{bias} of an estimator is the vector field $b$ on $\p$:
\begin{align}
	\forall \theta\in\p, \quad b(\theta) \triangleq \expectt{\Log_\theta(\hat\theta(y))},
\end{align}
where $\Log$ is the (Riemannian) logarithmic map, see~\eqref{eq:loganchored},~\eqref{eq:Logpempty}. An \emph{unbiased estimator} has a vanishing bias $b\equiv 0$.
\end{definition}

\subsection{Anchored synchronization}\label{subsec:anchoredcrb}

When anchors are provided, the rotation matrices $R_i$ for $i\in A$, $A\neq \emptyset$ are known. The parameter space then becomes $\p_A$~\eqref{eq:pA}, which is a Riemannian submanifold of $\p$. The synchronization problem is well-posed on $\p_A$, provided there is at least one anchor in each connected component of the measurement graph.
\change{
We set out to use Theorem 4 in~{\citep{crbsubquot}}, which states Cramér-Rao bounds for estimation problems on Riemannian submanifolds. To this end, we first define the $dN\times dN$ covariance matrix $C_A$ for anchored synchronization, then we introduce a projected FIM $F_A$, and we finally state the bound.}{Let us define the covariance matrix of an estimator for anchored synchronization:}
\begin{definition}[anchored covariance]
Following~\citep[eq.\,(5)]{crbsubquot}, the covariance matrix of an estimator $\hat \bfR$ mapping each possible set of measurements $H_{ij}$ to a point in $\p_A$, expressed w.r.t.\ the orthonormal basis~\eqref{eq:basisTRP} of $\T_{\bfR}\p$, is given by:
\begin{align}
	(C_A)_{ij,k\ell} & = \expectt{\innersmall{\Log_{\bfR}(\hat \bfR)}{\bfxi_{ik}} \cdot \innersmall{\Log_{\bfR}(\hat \bfR)}{\bfxi_{j\ell}}},
	\label{eq:CA}
\end{align}
where the indexing convention is the same as for the FIM. Of course, all $d\times d$ blocks $(i,j)$ such that either $i$ or $j$ or both are in $A$ are zero by construction. In particular, the variance of $\hat \bfR$ is the trace of $C_A$:
\begin{align}
	\trace\,C_A & = \expectt{\|\Log_{\bfR}(\hat \bfR)\|^2} = \expectt{\dist^2(\bfR, \hat \bfR)},
	\label{eq:varA}
\end{align}
where $\dist$ is the geodesic distance on $\p_A$~\eqref{eq:distanchored}.
\end{definition}

\add{Cramér-Rao bounds link this covariance matrix to the Fisher information matrix derived in the previous section through the following result, which is a statement of Theorem 4 in~{\citep{crbsubquot}} with adapted notation:}

\begin{theorem}[\add{anchored CRB}] \label{thm:CRBsubmanifold}
Given any unbiased estimator $\hat\bfR$ for synchronization on $\p_A$, at large SNR, the covariance matrix $C_A$~\eqref{eq:CA} and the Fisher information matrix $F$~\eqref{eq:Fgeneral} obey the matrix inequality (assuming at least one anchor in each connected component):
\begin{align}
	C_A \succeq F_A^\dagger - \frac{1}{3}\left( R_m(F_A^\dagger)F_A^\dagger + F_A^\dagger R_m(F_A^\dagger) \right),
	\label{eq:CRBsubmanifold}
\end{align}
where $F_A = P_A F P_A$ and $P_A$ is the orthogonal projector from $\T_{\bfR}\p$ to $\T_{\bfR}\p_A$, expressed w.r.t.\ the orthonormal basis~\eqref{eq:basisTRP}. The operator $R_m \colon \mathbb{R}^{dN\times dN} \to \mathbb{R}^{dN\times dN}$ involves the Riemannian curvature tensor of $\p_A$ and is detailed in Appendix~\ref{apdx:curvatureterms}.
\end{theorem}

The effect of $P_A$ is to set all rows and columns corresponding to anchored rotations to zero. Thus, we introduce the masked Laplacian $\mathcal{L}_A$:
\begin{align}
	(\mathcal{L}_A)_{ij} = \begin{cases} \mathcal{L}_{ij} & \textrm{ if } i,j \notin A, \\ 0 & \textrm{ otherwise.} \end{cases}
	\label{eq:LA}
\end{align}
Then, the projected FIM is simply:
\begin{align}
	F_A = \frac{1}{d}(\mathcal{L}_A\otimes I_d).
	\label{eq:FA}
\end{align}
The pseudoinverse of $F_A$ is given by $F_A^\dagger = d(\mathcal{L}_A^\dagger \otimes I_d)$, since for arbitrary matrices $A$ and $B$, it holds that $(A\otimes B)^\dagger = A^\dagger \otimes B^\dagger$~\cite[Fact~7.4.32]{bernstein2009matrix}. Notice that the rows and columns of $\mathcal{L}_A^\dagger$ corresponding to anchors are also zero. Theorem~\ref{thm:CRBsubmanifold} then yields the sought CRB:
\begin{align}
	C_A \succeq d(\mathcal{L}_A^\dagger \otimes I_d) + \textrm{curvature terms}.
\label{eq:CRBanchored}
\end{align}
\remove{We compute the curvature terms explicitly in Appendix~{\ref{apdx:curvatureterms}} and show that they are on the order of SNR$^{-2}$. They can hence be neglected for large SNR, which makes for easier formulas.} In particular, for $n=2$, the manifold $\p_A$ is flat and $d = 1$. Hence, the curvature terms vanish exactly ($R_m \equiv 0$) and the CRB reads:
\begin{align}
	C_A & \succeq \mathcal{L}_A^\dagger.
\label{eq:CRBanchoredplane}
\end{align}
For $n=3$, including the curvature terms as detailed in Appendix~\ref{apdx:curvatureterms} yields this CRB:
\begin{align}
	C_A & \succeq 3 \left(\mathcal{L}_A^\dagger - \frac{1}{4}\left(\ddiag(\mathcal{L}_A^\dagger)\mathcal{L}_A^\dagger + \mathcal{L}_A^\dagger\ddiag(\mathcal{L}_A^\dagger)\right)\right) \otimes I_3,
\label{eq:CRBanchoredspace}
\end{align}
where $\ddiag$ sets all off-diagonal entries of a matrix to zero. \change{The curvature terms are thus on the order of $\mathcal{O}(\lambda_\textrm{max}^2(\mathcal{L}_A^\dagger))$, which indeed becomes small at large SNR.}{At large SNR, that is, for small values of $\trace(\mathcal{L}_A^\dagger)$, the curvature terms hence effectively become negligible compared to the leading term.} For general $n$, neglecting curvature if $n\geq 3$, the variance is lower-bounded as follows:
\begin{align}
	\expectt{\dist^2(\bfR, \hat \bfR)} & \geq d^2\, \trace\,\mathcal{L}_A^\dagger,
\label{eq:varianceanchored}
\end{align}
where $\dist$ is as defined by~\eqref{eq:distanchored}. It also holds for each node $i$ that
\begin{align}
	\expectt{\dist^2(R_i, \hat R_i)} \geq d^2\, (\mathcal{L}_A^\dagger)_{ii}.
\label{eq:erroranchored}
\end{align}
\change{Interestingly, for synchronization on the group of translations $\Rn$, the CRB involves the same Laplacian-related matrix~{\citep{crbsubquot}}. Barooah and Hespanha give interpretations of this bound in terms of the resistance distance on the measurement graph when there is exactly one anchor~{\citep{barooah2007estimation}}. This analogy is notably useful in verifying important properties of the bound, such as its behavior when edges or nodes are added or deleted, etc.}{This leads to a useful interpretation of the CRB in terms of a resistance distance on the measurement graph, as depicted in Figure~{\ref{fig:graphanchored}}. Indeed, for a general setting with one or more anchors, it can be checked that~{\citep{bouldin1973pseudo}}
\protect\begin{align}
	\mathcal{L}_A^\dagger & = (J_A (\mathcal{D} - \mathcal{A}) J_A)^\dagger = (J_A(I_N - \mathcal{D}^{-1}\mathcal{A})J_A)^\dagger \mathcal{D}^{-1},
	\label{eq:randomwalkanchored}
\end{align}
where $J_A$ is a diagonal matrix such that $(J_A)_{ii} = 1$ if $i\in A$ and $(J_A)_{ii} = 0$ otherwise.
It is well-known, e.g.\ from~{\citep[\S\,1.2.6]{doyle2000random}}, that in the first factor of the right-hand side, $((J_A(I_N - \mathcal{D}^{-1}\mathcal{A})J_A)^\dagger)_{ii}$ is the average number of times a random walker starting at node $i$ on the graph with transition probabilities $\mathcal{D}^{-1}\mathcal{A}$ will be at node $i$ before hitting an anchor. This number is small if node $i$ is strongly connected to anchors. In the CRB~\eqref{eq:erroranchored} on node $i$, $(\mathcal{L}_A^\dagger)_{ii}$ is thus the ratio between this anchor-connectivity measure and the overall amount of information available about node $i$ directly, namely $\mathcal{D}_{ii} = \sum_{j \in V_i} w_{ij}$.}

\begin{figure}[t]
\centering
\includegraphics[width=\textwidth]{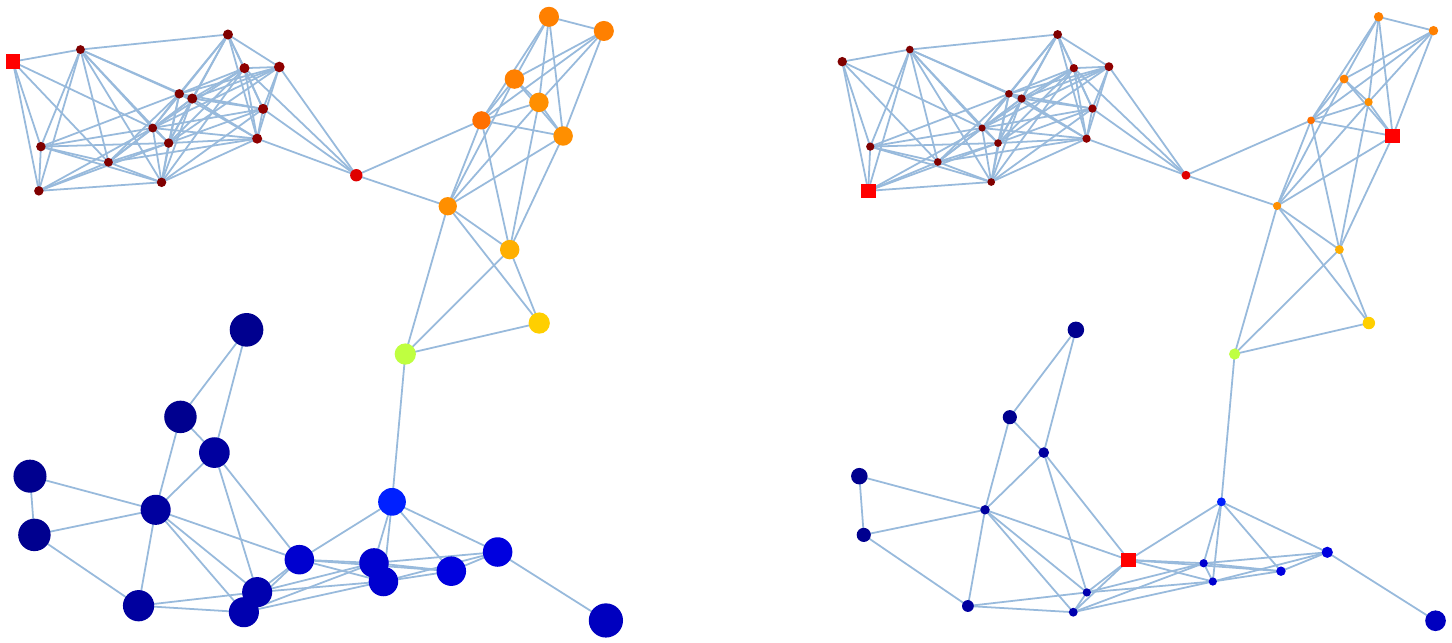}
\caption{The Cramér-Rao bound for \emph{anchored} synchronization~\eqref{eq:erroranchored} limits how well each individual rotation can be estimated. The two identical synchronization graphs above illustrate the effect of anchors. All edges have the same weight (i.i.d.\ noise). Anchors are red squares. Unknown rotations are round nodes\add{ colored according to the second eigenvector of $\mathcal{L}$ to bring out the clusters}. The area of node $i$ is proportional to the lower-bound on the average error for this node $\expectt{\dist^2(R_i, \hat R_i)}$. On the left, there is only one anchor in the upper-left cluster. Hence, nodes in the lower-left cluster, which are separated from the anchor by two bottlenecks, will be harder to estimate accurately than in the situation on the right, where there is one anchor for each cluster. Node positions in the picture are irrelevant.}
\label{fig:graphanchored}
\end{figure}

\subsection{Anchor-free synchronization}\label{subsec:anchorfreecrb}

When no anchors are provided, the global rotation ambiguity leads to the equivalence relation~\eqref{eq:equivrelation} on $\p$, which in turn leads to work on the Riemannian quotient parameter space $\p_\emptyset$~\eqref{eq:pempty}. The synchronization problem is well-posed on $\p_\emptyset$ as long as the measurement graph is connected, which we always assume in this work.
\change{
We set out to use Theorem 5 in~{\citep{crbsubquot}}, which states Cramér-Rao bounds for estimation problems on Riemannian quotient manifolds. To this end, we first define the $dN\times dN$ covariance matrix $C_\emptyset$ for anchor-free synchronization, then we state the bound.}{Let us define the covariance matrix of an estimator for anchor-free synchronization:}
\begin{definition}[anchor-free covariance]
Following~\citep[eq.\,(20)]{crbsubquot}, the covariance matrix of an estimator $[\hat \bfR]$ mapping each possible set of measurements $H_{ij}$ to a point in $\p_\emptyset$ (that is, to an equivalence class in $\p$), expressed w.r.t.\ the orthonormal basis~\eqref{eq:basisTRP} of $\T_{\bfR}\p$, is given by:
\begin{align}
	(C_\emptyset)_{ij,k\ell} & = \expectt{\innersmall{\bfxi}{\bfxi_{ik}} \cdot \innersmall{\bfxi}{\bfxi_{j\ell}}}, \textrm{ with} \nonumber\\
	\bfxi & = (\D\pi(\bfR)|_{\HH_{\bfR}})^{-1}[\Log_{[\bfR]}([\hat \bfR])].
	\label{eq:Cempty}
\end{align}
That is, $\bfxi$ (the random error vector) is the shortest horizontal vector such that $\Exp_{\bfR}(\bfxi) \in [\hat \bfR]$~\eqref{eq:loganchored}. We used the same indexing convention as for the FIM. In particular, the variance of $[\hat \bfR]$ is the trace of $C_\emptyset$:
\begin{align}
	\trace\,C_\emptyset & = \expectt{\|\Log_{[\bfR]}([\hat \bfR])\|^2} = \expectt{\dist^2([\bfR], [\hat \bfR])},
	\label{eq:varempty}
\end{align}
where $\dist$ is the geodesic distance on $\p_\emptyset$~\eqref{eq:distanchorfree}.
\end{definition}

\add{Theorem 5 in~{\citep{crbsubquot}}, stated here with adapted notation, links this covariance matrix to the Fisher information matrix as follows:}

\begin{theorem}[\add{anchor-free CRB}] \label{thm:CRBquotientmanifold}
Given any unbiased estimator $[\hat\bfR]$ for synchronization on $\p_\emptyset$, at large SNR, the covariance matrix $C_\emptyset$~\eqref{eq:Cempty} and the Fisher information matrix $F$~\eqref{eq:Fgeneral} obey the matrix inequality (assuming the measurement graph is connected):
\begin{align}
	C_\emptyset \succeq F^\dagger - \frac{1}{3}\left( R_m(F^\dagger)F^\dagger + F^\dagger R_m(F^\dagger) \right),
	\label{eq:CRBquotientmanifold}
\end{align}
where $R_m \colon \mathbb{R}^{dN\times dN} \to \mathbb{R}^{dN\times dN}$ involves the Riemannian curvature tensor of $\p_\emptyset$ and is detailed in Appendix~\ref{apdx:curvatureterms}.
\end{theorem}

Theorem~\ref{thm:CRBquotientmanifold} then yields the sought CRB:
\begin{align}
	C_\emptyset \succeq d(\mathcal{L}^\dagger \otimes I_d) + \textrm{curvature terms}.
\label{eq:CRBanchorfree}
\end{align}
We compute the curvature terms explicitly in Appendix~\ref{apdx:curvatureterms} and show they can be neglected for \remove{(reasonably) }large \remove{enough }SNR. In particular, for $n=2$, the manifold $\p_\emptyset$ is flat and $d = 1$. Hence:
\begin{align}
	C_\emptyset \succeq \mathcal{L}^\dagger.
\label{eq:CRBanchorfreeplane}
\end{align}
For $n=3$, the curvature terms are the same as those for the anchored case, with an additional term that decreases as $1/N$. For (not so) large $N$ then, the bound~\eqref{eq:CRBanchoredspace} is a good bound for $n=3$, anchor-free synchronization. For general $n$, neglecting curvature for $n\geq 3$, the variance is lower-bounded as follows:
\begin{align}
	\expectt{\dist^2([\bfR], [\hat \bfR])} & \geq d^2\, \trace\,\mathcal{L}^\dagger,
\label{eq:varianceanchorfree}
\end{align}
where $\dist$ is as defined by~\eqref{eq:distanchorfree}.

For the remainder of this section, we work out an interpretation of~\eqref{eq:CRBanchorfree}. This matrix inequality entails that, for all $x\in\mathbb{R}^{dN}$ (neglecting curvature terms if needed):
\begin{align}
	x\transpose C_\emptyset x \geq d\,x\transpose(\mathcal{L}^\dagger\otimes I_d)x.
\label{eq:inequalityanchorfree}
\end{align}
As both the covariance and the FIM correspond to positive semidefinite operators on the horizontal space $\HH_{\bfR}$, this is really only meaningful when $x$ is the vector of coordinates of a horizontal vector $\bfeta = (\eta_1, \ldots, \eta_N) \in \HH_{\bfR}$. \add{We emphasize that this restriction implies that the anchor-free CRB, as it should, only conveys information about \emph{relative} rotations. It does not say anything about singled-out rotations in particular.} Let $e_i$, $e_j$ denote the $i$-th and $j$-th columns of the identity matrix $I_N$ and let $e_k$ denote the $k$-th column of $I_d$. We consider $x = (e_i-e_j)\otimes e_k$, which corresponds to the zero horizontal vector $\bfeta$ except for $\eta_i = R_iE_k$ and $\eta_j = -R_jE_k$, with $E_k\in\son$ the $k$-th element of the orthonormal basis of $\son$ picked as in~\eqref{eq:basisSOt}. By definition of $C_\emptyset$ and of the error vector $\bfxi = (R_1\Omega_1, \ldots, R_N\Omega_N) \in \HH_{\bfR}$~\eqref{eq:Cempty},
\begin{align}
	x\transpose C_\emptyset x & = \expectt{\inner{\bfxi}{\eta}^2} = \expectt{\inner{\Omega_i-\Omega_j}{E_k}^2}.
\end{align}
On the other hand, we have
\begin{align}
	d\,x\transpose (\mathcal{L}^\dagger\otimes I_d) x & = d\,(e_i-e_j)\transpose \mathcal{L}^\dagger (e_i-e_j).
\end{align}
These two last quantities are related by inequality~\eqref{eq:inequalityanchorfree}. Summing for $k = 1\ldots d$ on both sides of this inequality, we find:
\begin{align}
	\expectt{\sqfrobnorm{\Omega_i-\Omega_j}} & \geq d^2\,(e_i-e_j)\transpose \mathcal{L}^\dagger (e_i-e_j).
\end{align}
Now remember that the error vector $\bfxi$~\eqref{eq:Cempty} is the shortest horizontal vector such that $\Exp_{\bfR}(\bfxi) \in [\hat \bfR]$. Without loss of generality, we may assume that $\hat \bfR$ is aligned such that $\Exp_{\bfR}(\bfxi) = \hat \bfR$. Then, $\hat R_i = R_i\exp(\Omega_i)$ for all $i$. It follows that
\begin{align}
	\hat R_i^{}\hat R_j\transpose & = R_i^{}\exp(\Omega_i^{})\exp(-\Omega_j^{})R_j\transpose, \quad \textrm{ hence} \\
	\dist^2(R_i^{}R_j\transpose, \hat R_i^{}\hat R_j\transpose) & = \sqfrobnorm{\log\big(\exp(\Omega_i)\exp(-\Omega_j)\big)}.
\end{align}
For commuting $\Omega_i$ and $\Omega_j$---which is always the case for $n=2$---we have
\begin{align}
	\log\big(\exp(\Omega_i)\exp(-\Omega_j)\big) = \Omega_i-\Omega_j.
\end{align}
For $n\geq 3$, this still approximately holds \change{at high enough SNR}{in small error regimes} (that is, for small enough $\Omega_i, \Omega_j$), by the Baker-Campbell-Hausdorff formula. Hence,
\begin{align}
	\expectt{\dist^2(R_i^{}R_j\transpose, \hat R_i^{}\hat R_j\transpose)} & \approx \expectt{\sqfrobnorm{\Omega_i-\Omega_j}} \geq d^2\,(e_i-e_j)\transpose \mathcal{L}^\dagger (e_i-e_j).
	\label{eq:erroranchorfree}
\end{align}
The quantity $\trace(\mathcal{D}) \cdot (e_i-e_j)\transpose \mathcal{L}^\dagger (e_i-e_j)$ is sometimes called the squared Euclidean commute time distance (ECTD)~\citep{saerens2004principal} between nodes $i$ and $j$. It is also known as the electrical resistance distance. \add{For a random walker on the graph with transition probabilities $\mathcal{D}^{-1}\mathcal{A}$, this quantity is the average commute time distance, that is, the number of steps it takes on average, starting at node $i$, to hit node $j$ then node $i$ again. The right-hand side of~{\eqref{eq:erroranchorfree}} is thus inversely proportional to the quantity and quality of information linking these two nodes. It decreases whenever the number of paths between them increases or whenever an existing path is made more informative, i.e., weights on that path are increased.}

Still in~\citep{saerens2004principal}, it is shown in Section~5 how principal component analysis (PCA) on $\mathcal{L}^\dagger$ can be used to embed the nodes in a low dimensional subspace such that the Euclidean distance separating two nodes is similar to the ECTD separating them in the graph. For synchronization, such an embedding naturally groups together nodes whose relative rotations can be accurately estimated, as depicted in Figure~\ref{fig:ECTD}.

\remove{The CRB for anchor-free synchronization thus sets a lower-bound on how well one can estimate the relative rotation between $R_i$ and $R_j$. The bound reflects how strongly the corresponding nodes are connected in the information-weighted measurement graph. }

\begin{figure}[t]
\centering
\includegraphics[width=\textwidth]{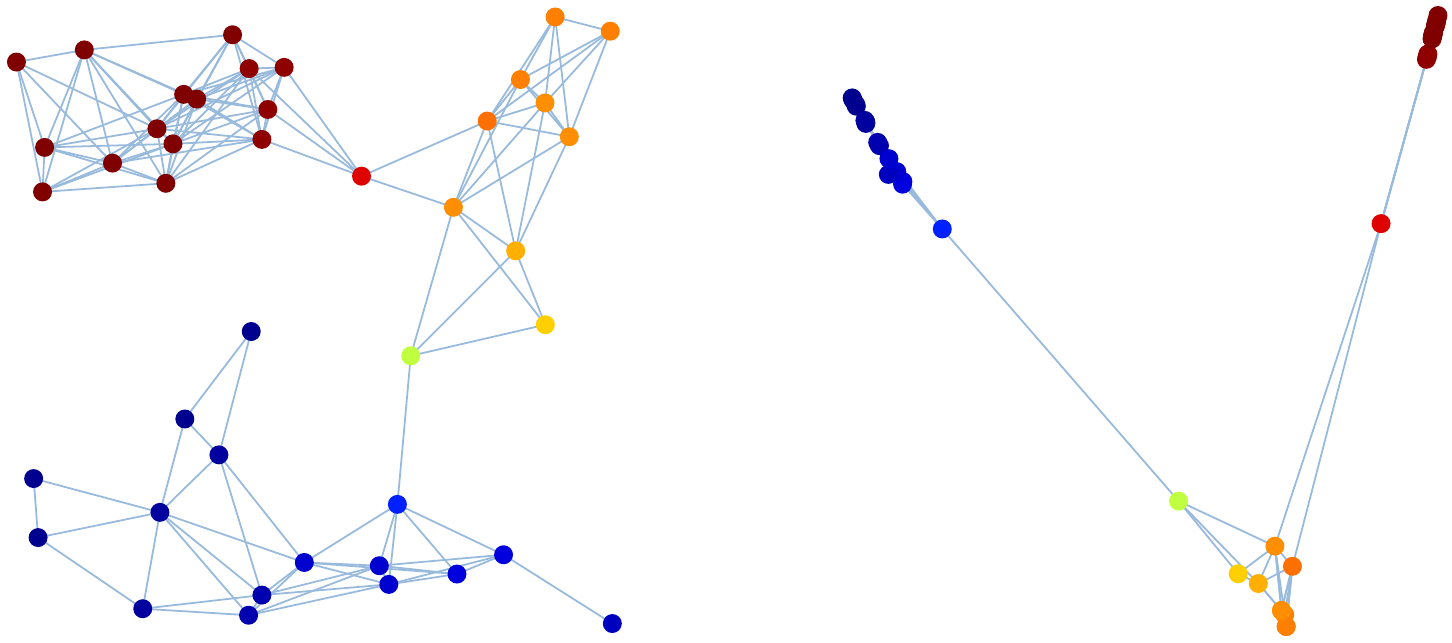}
\caption{\changesafe{The Fisher information matrix for \emph{anchor-free} synchronization is essentially the Laplacian of the measurement graph, whose edges have weights proportional to the amount of information these edges contain~{\eqref{eq:Fgeneral}}. As a result, the Cramér-Rao bound limits the accuracy to which}{The Cramér-Rao bound for \emph{anchor-free} synchronization~{\eqref{eq:erroranchorfree}} limits how well} the relative rotation between two nodes can be estimated, in proportion to the Euclidean commute time distance (ECTD) separating them in the graph. Left: each node in the synchronization graph corresponds to a rotation to estimate and each edge corresponds to a measurement of relative rotation. Noise affecting the measurements is i.i.d., hence all edges have the same weight. Nodes are colored according to the second eigenvector of $\mathcal{L}$ (the Fiedler vector). Node positions are irrelevant. Right: ECTD-embedding of the same graph in the plane, such that the distance between two nodes $i$ and $j$ in the picture is mostly proportional to the ECTD separating them, which is essentially a lower-bound on $\expect\{\dist^2(R_i^{}R_j\transpose, \hat R_i^{}\hat R_j\transpose)\}^{1/2}$. In other words: the closer two nodes are, the better their relative rotation can be estimated. \add{Notice that the node colors correspond to the horizontal coordinate in the right picture.} See Section~{\ref{subsec:visualization}}.}
\label{fig:ECTD}
\end{figure}

%% file: analysis.tex
\section{Comments on, and consequences of the CRB}
\label{sec:analysis}

\add{So far, we derived the Cramér-Rao bounds for synchronization in both anchored and anchor-free settings. These bounds enjoy a rich structure and lend themselves to useful interpretations, as for the random walk perspective for example. In this section, we start by hinting that the CRB might be achievable, making it all the more relevant. We then stress the limits of validity of the CRB's, namely the large SNR proviso. Visualization tools are proposed to assist in graph analysis. Finally, we focus on anchor-free synchronization and comment upon the role of the Fiedler value of a measurement graph, the synchronizability of Erd\"os-R\'enyi random graphs and the remarkable resilience to outliers of synchronization.}

\subsection{The CRB's might be achievable} \label{subsec:achievable}

\add{Figure~{\ref{fig:plots}} compares the CRB~{\eqref{eq:CRBanchoredspace}} against three existing synchronization algorithms on synthetic data. The MLE method~{\citep{boumal2013MLE}}, which is a proxy for the maximum likelihood estimator for synchronization, seems to achieve the CRB as the SNR increases. There is no proof yet that this is indeed the case, but nevertheless the small gap between the empirical variance of the MLE and the CRB suggests that studying the CRB is relevant to understand synchronization.}

\begin{figure}
	\centering
	\mbox{
		\subfigure{\includegraphics[width=7cm]{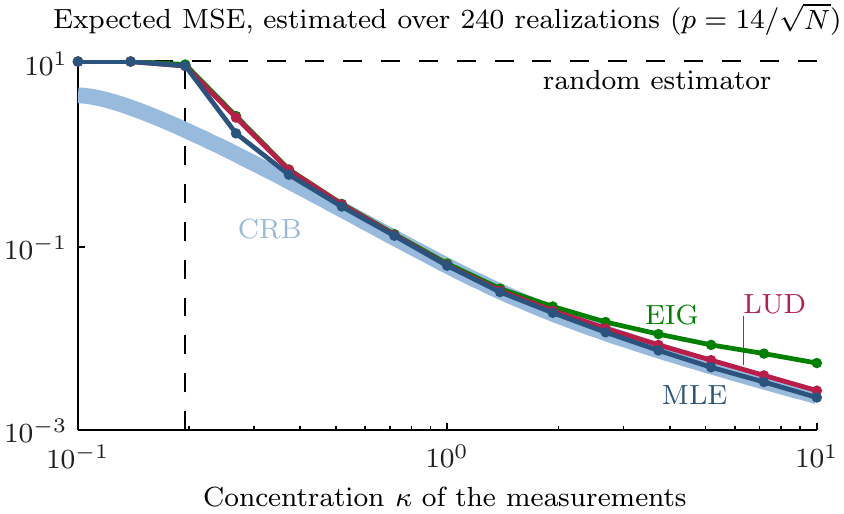}} \quad
		\subfigure{\includegraphics[width=7cm]{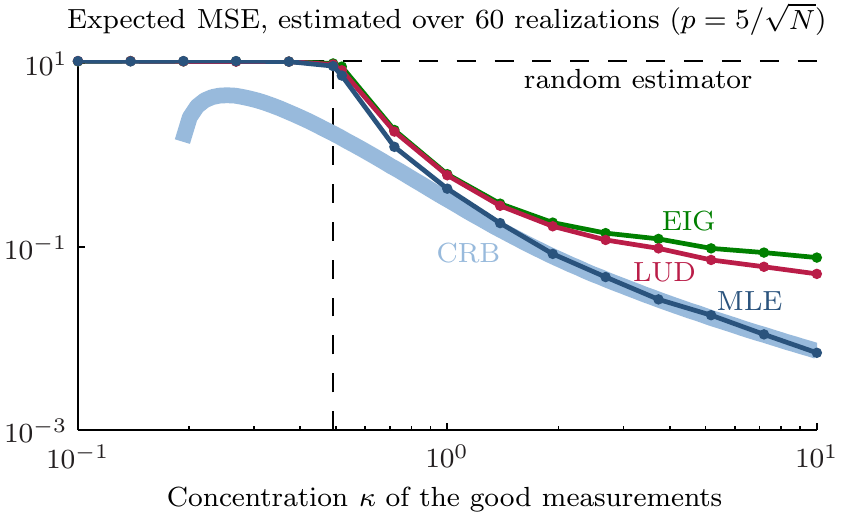}}
	}
	\caption{\add{Comparison of the CRB~{\eqref{eq:CRBanchoredspace}} with the (empirical) expected mean squared error \protect{$\expectt{\dist^2(\bfR, \hat\bfR)/(N-1)}$} (MSE) of known estimators for synchronization of $N = 400$ rotations in $\SOt$ with a complete measurement graph and one anchor. Left: measurements follow the Langevin + outlier model, with 30\% of outliers. Good measurements have variable concentration $\kappa$ on the horizontal axis. The eigenvector method EIG~{\citep{singer2010angular}}, the least unsquared deviation method LUD~{\citep{wang2012LUD}} and the (proxy for the) maximum likelihood estimator MLE~{\citep{boumal2013MLE}} are depicted. MLE has perfect knowledge of the noise model and appears to reach the CRB. LUD performs excellently without specific knowledge of the noise model. Right: same scenario with 75\% of outliers. MLE still appears to reach the CRB. Both: the horizontal dashed line indicates the MSE reached by a random estimator. The vertical dashed line indicates the (theoretically predicted) phase transition point at which the eigenvector method starts performing better than randomly.}}
	\label{fig:plots}
\end{figure}

\subsection{The CRB is an asymptotic bound} \label{subsec:asymptoticbound}

Intrinsic Cramér-Rao bounds are asymptotic bounds, that is, they are meaningful for small errors. \add{This stems from two reasons. First, when the parameter space is curved, only the leading order curvature term has been computed~{\citep{smith2005covariance}}. This induces a Taylor truncation error which restricts the validity of the CRB's to small error regimes. Second, the parameter spaces $\p_A$ and $\p_\emptyset$ are compact, hence there is an upper-bound on the variance of any estimator. The CRB is unable to capture such a global feature because it is derived under the assumption that the logarithmic map $\Log$ is globally invertible, which compactness prevents.} Hence, for arbitrarily low SNR, the CRB without curvature terms will predict an arbitrarily large variance and will be violated \add{(this does not show on Figure~{\ref{fig:plots}} since the CRB depicted includes curvature terms, which in this case make the CRB go to zero at low SNR)}. As a means to locate the point at which the CRB certainly stops making sense, 
consider the problem of estimating a rotation matrix $R\in\SOn$ based on a measurement $Z\in\SOn$ of $R$, and compute the variance of the (unbiased) estimator $\hat R(Z) := Z$ when $Z$ is uniformly random, i.e., when no information is available.

Define $V_n = \expectt{\dist^2(Z, R)}$ for $Z\sim \uniform(\SOn)$. A computation using Weyl's formula yields:
\begin{align}
	V_2 & = \frac{1}{2\pi}\int_{-\pi}^\pi \smallfrobnorm{\log(Z\transpose R)}^2 \,\mathrm{d}\theta = \frac{1}{2\pi}\int_{-\pi}^\pi 2\theta^2 \,\mathrm{d}\theta = \frac{2\pi^2}{3}, &
	V_3 & = \frac{2\pi^2}{3} + 4.
	\label{eq:randomvariance}
\end{align}
A reasonable upper-bound on the variance of an estimator should be $N'V_n$, where $N'$ is the number of independent rotations to estimate ($N-1$ for anchor-free synchronization, $N-|A|$ for anchored synchronization)---see Figure~\ref{fig:plots}. A CRB larger than this should be disregarded.

\remove{Neglecting curvature terms---which we showed is exact for $n=2$ and legitimate for $n\geq 3$ at large SNR---the Cramér-Rao bounds yield simple lower-bounds on the variance of unbiased estimators for synchronization. The right-hand sides of these bounds could in turn be used as cost function in an effort to optimize network structure or anchor placement. We take a closer look at the anchor-free bound to determine what structural properties affect it.}

\subsection{Visualization tools} \label{subsec:visualization}

\add{
In deriving the anchor-free bounds for synchronization, we established that a lower bound on \protect$$\expectt{\dist^2(R_i^{}R_j\transpose, \hat R_i^{}\hat R_j\transpose)}$$ is proportional to the quantity $(e_i-e_j)\transpose \mathcal{L}^\dagger (e_i-e_j)$. Of course, this analysis also holds for anchored graphs. Here, we detail how Figure~{\ref{fig:ECTD}} was produced following a PCA procedure~{\citep{saerens2004principal}} and show how this translates for anchored graphs, as depicted in Figure~{\ref{fig:graphanchored2}}.

We treat both anchored and anchor-free scenarios, thus allowing $A$ to be empty in this paragraph. Let $\mathcal{L}_A = VDV\transpose$ be an eigendecomposition of $\mathcal{L}_A$, such that $V$ is orthogonal and $D = \diag(\lambda_1, \ldots, \lambda_N)$. Let $X = (D^\dagger)^{1/2} V\transpose$, an $N\times N$ matrix with columns $x_1, \ldots, x_N$. Assume without loss of generality that the eigenvalues and eigenvectors are ordered such that the diagonal entries of $D^\dagger$ are decreasing. Then,
\protect\begin{align}
	(e_i-e_j)\transpose \mathcal{L}_A^\dagger (e_i-e_j) & = (e_i-e_j)\transpose VD^\dagger V\transpose (e_i-e_j) = \|x_i - x_j\|^2.
\end{align}
Thus, embedding the nodes at positions $x_i$ realizes the ECTD in $\mathbb{R}^N$. Anchors, if any, are placed at the origin. An optimal embedding, in the sense of preserving the ECTD as well as possible, in a subspace of dimension $k < N$ is obtained by considering $\tilde X$: the first $k$ rows of $X$. The larger the ratio $\sum_{\ell = 1}^k \lambda_\ell^\dagger / \trace(D^\dagger)$, the better the low-dimensional embedding captures the ECTD.

In the presence of anchors, if $j\in A$, then $\mathcal{L}_A^\dagger e_j = 0$ and $(e_i-e_j)\transpose \mathcal{L}_A^\dagger (e_i-e_j) = (\mathcal{L}_A^\dagger)_{ii} = \|x_i\|^2 \approx \|\tilde x_i\|^2$. Hence, the embedded distance to the origin indicates how well a specific node can be estimated.

In practice, this embedding can be produced by computing the $m+k$ eigenvectors of $\mathcal{L}_A$ with smallest eigenvalue, where $m = \max(1, |A|)$ is the number of zero eigenvalues to be discarded (assuming a connected graph). This computation can be conducted efficiently if the graph is structured, e.g., sparse.
}

\begin{figure}[t]
\centering
\includegraphics[width=\textwidth]{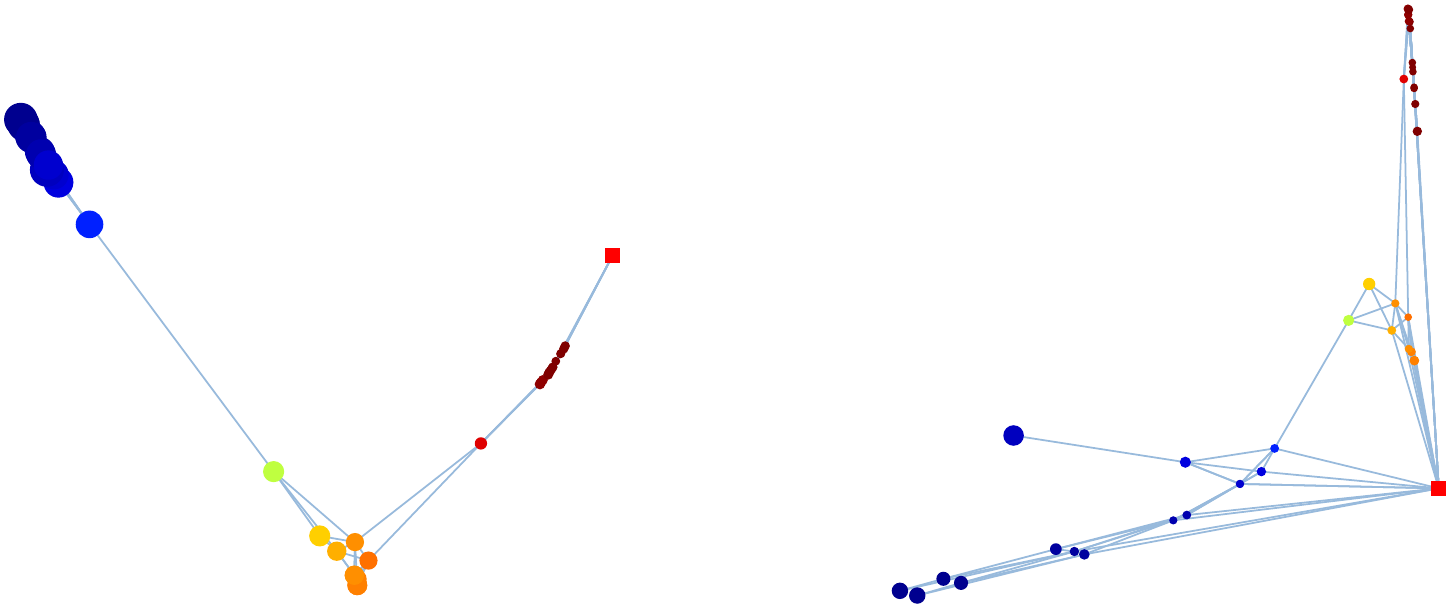}
\caption{\add{The visualization tool described in Section~\ref{sec:analysis}, applied here to the anchored synchronization tasks from Figure~\ref{fig:graphanchored}, produces low-dimensional embeddings of synchronization graphs such that the distance between two nodes is large if their relative rotation is hard to estimate, and their distance to the origin (the anchors: red squares) is large if their individual rotation is hard to estimate.}}
\label{fig:graphanchored2}
\end{figure}

\subsection{A larger Fiedler value is better} \label{subsec:largeFiedler}

We now focus on anchor-free synchronization. At large SNR, the anchor-free CRB~\eqref{eq:varianceanchorfree} normalized by the number of independent rotations $N-1$ reads:
\begin{align}
	\expectt{\MSE} \triangleq \expectt{\frac{1}{N-1}\dist^2([\bfR], [\hat \bfR])} \geq \frac{d^2}{N-1}\,\trace(\mathcal{L}^\dagger),
	\label{eq:varianceLB}
\end{align}
where $\expectt{\MSE}$ as defined is the \emph{expected mean squared error} of an unbiased estimator $[\hat \bfR]$. This expression shows the limiting role of the trace of the pseudoinverse of the information-weighted Laplacian $\mathcal{L}$~\eqref{eq:laplacian} of the measurement graph. This role has been established before for other synchronization problems for simpler groups and simpler noise models~\citep{howard2010estimation}. We now shed some light on this result by stating a few elementary consequences of it. Let
\begin{align}
	0 = \lambda_1 < \lambda_2 \leq \cdots \leq \lambda_N
\end{align}
denote the eigenvalues of $\mathcal{L}$, where $\lambda_2 > 0$ means the measurement graph is assumed connected.

The right-hand side of~\eqref{eq:varianceLB} in terms of the $\lambda_i$'s is given by:
\begin{align}
	\frac{d^2}{N-1}\,\trace(\mathcal{L}^\dagger) & = \frac{d^2}{N-1}\,\sum_{i=2}^N \frac{1}{\lambda_i} \leq \frac{d^2}{\lambda_2}.
\end{align}
The second eigenvalue $\lambda_2$ is known as the  \emph{Fiedler value} (or \emph{algebraic connectivity}) of the information-weighted measurement graph. It is well known that the Fiedler value is low in the presence of bottlenecks in the graph and high in the presence of many, heavy spanning trees. The latter equation translates in the following intuitive statement: by increasing the Fiedler value of the measurement graph, one can force \change{the right-hand side of the CRB~{\eqref{eq:varianceLB}} down}{a lower CRB}. Not surprisingly then, expander graphs are ideal for synchronization, since, by design, their Fiedler value $\lambda_2$ is bounded away from zero while simultaneously being sparse~\citep{hoory2006expander}.

Notice that the Fiedler vector has zero mean (it is orthogonal to $\mathds{1}_N$) and hence describes the horizontal vectors of maximum variance. It is thus also the first axis of the right plot in Figure~\ref{fig:ECTD}.

\subsection{$\trace(\mathcal{L}^\dagger)$ plays a limiting role in synchronization} \label{subsec:traceLdagger}

We continue to focus on anchor-free synchronization. The quantity $\trace(\mathcal{L}^\dagger)$ appears naturally in CRB's for synchronization problems on groups~\citep{howard2010estimation,crbsubquot}. For complete graphs and constant weight $w$, $\trace(\mathcal{L}^\dagger) = \frac{N-1}{wN}$. Then, by~\eqref{eq:varianceLB},
\begin{align}
	\expectt{\MSE} \geq \frac{d^2}{wN}.
	\label{eq:MSEcomplete}
\end{align}
If the measurement graph is sampled from a distribution of random graphs, $\trace(\mathcal{L}^\dagger)$ becomes a random variable. We feel that the study of this random variable for various families of random graph models, such as Erd\"os-R\'enyi, small-world or scale-free graphs~\citep{jamakovic2007relationship} is a question of interest, probably best addressed using the language of random matrix theory.

Let us consider Erd\"os-R\'enyi graphs $G_{N,q}$ with $N$ nodes and edge density $q\in(0,1)$, that is, graphs such that any edge is present with probability $q$, independently from the other edges. Let all the edges have equal weight $w$. Let $\mathcal{L}_{N,q}$ be the Laplacian of a $G_{N,q}$ graph. The expected Laplacian is $\expectt{\mathcal{L}_{N,q}} = wq(NI_N - \mathds{1}_{N\times N})$, which has eigenvalues $\lambda_1=0, \lambda_2=\cdots=\lambda_N = Nwq$. Hence, $\trace(\expectt{\mathcal{L}_{N,q}}^\dagger) = \frac{N-1}{N}\frac{1}{wq}$. A more useful statement can be made using~\citep[Thm.\,1.4]{bryc2006spectral} and~\citep[Thm.\,2]{ding2010spectral}. These theorems state that, asymptotically as $N$ grows to infinity, all eigenvalues of $\mathcal{L}_{N,q}/N$ converge to $wq$ (except of course for one zero eigenvalue). Consequently (details omitted),
\begin{align}
	\lim_{N\to\infty} \trace(\mathcal{L}_{N,q}^\dagger) = \frac{1}{wq}  \textrm{ (in probability).}
	\label{eq:trpinvlapER}
\end{align}
For large $N$, we use the approximation $\trace(\mathcal{L}_{N,q}^\dagger) \approx 1/wq$. Then, by~\eqref{eq:varianceLB}, for large $N$ we have:
\begin{align}
	\expectt{\MSE} \gtrsim \frac{d^2}{wqN}.
	\label{eq:MSEerdosrenyi}
\end{align}
Notice how for fixed measurement quality $w$ and density $q$, the lower-bound on the expected MSE decreases with the number $N$ of rotations to estimate.

\subsection{Synchronization can withstand many outliers} \label{subsec:withstandoutliers}

Consider the Langevin with outliers distribution from Example~\ref{ex:langevinoutliers}, where (on average) a fraction $1-p$ of measurements are sampled uniformly at random, i.e., they are outliers. The information weight $w(p) = \alpha_{n}(\kappa, p)$ for some fixed concentration $\kappa>0$ is given by equations~\eqref{eq:alphatwokp} and~\eqref{eq:alphathreekp} for $n=2$ and 3 respectively. A Taylor expansion around $p=0$ shows that
\begin{align}
	w(p) & = a_{n,\kappa} p^2 + \mathcal{O}(p^3)
\end{align}
for some positive constant $a_{n,\kappa}$. Then, for $p \ll 1$, building upon~\eqref{eq:MSEcomplete} for complete graphs with i.i.d.\ measurements we get:
\begin{align}
	\expectt{\MSE} & \gtrsim \frac{d^2}{a_{n,\kappa} p^2N}.
\end{align}
If one needs to get the right-hand side of this inequality down to a tolerance $\varepsilon^2$, the probability $p$ of a measurement not being an outlier needs to be at least as large as:
\begin{align}
	p_\varepsilon \triangleq \frac{d}{\sqrt{a_{n,\kappa}}\varepsilon} \frac{1}{\sqrt{N}}.
\end{align}
\remove{If $p_\varepsilon > 1$, synchronization cannot be solved to the desired accuracy. }The $1/\sqrt N$ factor is the most interesting: it establishes that as the number of nodes increases, synchronization can withstand a larger fraction of outliers.

This result is to be put in perspective with the bound in~\citep[eq.\,(37)]{singer2010angular} for $n=2$, $\kappa = \infty$, where it is shown that as soon as $p>1/\sqrt{N}$, there is enough information in the measurements (on average) for their eigenvector method to do better than random synchronization. It is also shown there that, as $p^2N$ goes to infinity, the correlation between the eigenvector estimator and the true rotations goes to 1. Similarly, we see that as $p^2N$ increases to infinity, the right-hand side of the CRB goes to zero. Our analysis further shows that the role of $p^2N$ is tied to the problem itself (not to a specific estimation algorithm), and remains the same for $n > 2$ and in the presence of Langevin noise on the good measurements.

Building upon~\eqref{eq:trpinvlapER} for Erd\"os-R\'enyi graphs with $N$ nodes and $M$ edges, we define $p_\varepsilon$ as:
\begin{align}
	p_\varepsilon \triangleq \frac{d}{\sqrt{a_{n,\kappa}}\varepsilon} \sqrt{\frac{N}{2M}}.
	\label{eq:pepsi}
\end{align}

To conclude this remark, we provide numerically computable expressions for $a_{n,\kappa}$, $n=2$ and 3 and give an example:
\begin{align}
	a_{2,\kappa} & = \frac{\kappa^2}{\pi c_2^2(\kappa)} \int_0^\pi (1-\cos 2\theta) \exp(4\kappa\cos\theta) \mathrm{d}\theta, \\[2mm]
	a_{3,\kappa} & = \frac{\kappa^2 e^{2\kappa}}{\pi c_3^2(\kappa)} \int_0^\pi (1-\cos 2\theta) (1-\cos\theta) \exp(4\kappa\cos\theta) \mathrm{d}\theta.
\end{align}
As an example, we generate an Erd\"os-R\'enyi graph with $N = 2500$ nodes and edge density of 60\% for synchronization of rotations in $\SOt$ with i.i.d.\ noise following a $\langout(I, \kappa=7, p)$. The CRB~\eqref{eq:varianceLB}, which requires complete knowledge of the graph to compute $\trace(\mathcal{L}^\dagger)$, tells us that we need $p \geq 2.1\%$ to reach an accuracy level of $\varepsilon = 10^{-1}$ (for comparison, $\varepsilon^2$ is roughly 1000 times smaller than $V_3$~\eqref{eq:randomvariance}). The simple formula~\eqref{eq:pepsi}, which can be computed quickly solely based on the graph statistics $N$ and $M$, yields $p_\varepsilon = 2.2\%$.

%% file: conclusions.tex
\section{Conclusions and perspectives}
\label{sec:conclusions}

In this work, we considered the problem of estimating a set of rotations $R_i\in\SOn$ based on pairwise measurements of relative rotations $R_i^{}R_j\transpose$. We provided a framework to study synchronization as estimation on manifolds for arbitrary $n$ under a large family of noise models. We established formulas for the Fisher information matrix and associated Cramér-Rao bounds of synchronization and provided interpretation and visualization tools for them in both the anchored and anchor-free scenarios. In the analysis of these bounds, we notably pointed out the high robustness of synchronization against random outliers.

Because of the crucial role of the pseudoinverse of the Laplacian $\mathcal{L}^\dagger$ of weighted graphs (and their traces) in the CRB's we established, it would be interesting to study efficient methods to compute such objects, see e.g.~\citep{ho2005pseudo,lin2009fast}.
Likewise, exploring the distribution of $\trace(\mathcal{L}^\dagger)$ seen as a random variable for various models of random graphs should bring some insight as to which networks are naturally easy to synchronize. Expander graphs already emerge as good candidates.

The Laplacian of the measurement graph plays the same role in bounds for synchronization of rotations as for synchronization of translations. Carefully checking the proof given in the present work, it is reasonable to speculate that the Laplacian would appear similarly in CRB's for synchronization on any Lie group, as long as we assume independence of noise affecting different measurements and some symmetry in the noise distribution. Such a generalization would in particular yield CRB's for synchronization on the special Euclidean group of rigid body motions, $\mathbb{R}^3 \rtimes \SOt$.

In other work, we leverage the formulation of synchronization as an estimation problem on manifolds to propose maximum likelihood estimators for synchronization~\citep{boumal2013MLE}. Such approaches result in optimization problems on the parameter manifolds whose geometries we described here. By executing the derivations with the Langevin + outliers noise model, this leads to naturally robust synchronization algorithms.

%% file: densitynorm.tex
\section{Langevin density normalization}
\label{apdx:coefficients}

This appendix presents the derivation of the normalization coefficient $c_4(\kappa)$~\eqref{eq:c4} that appears in the Langevin probability density function~\eqref{eq:langevin}. \add{In doing so, we use Weyl's integration formulas. This method applies well to compute the other integrals we need in this paper so that this appendix can be seen as an example.}

Recall that the coefficient $c_n(\kappa)$ is given by~\eqref{eq:cn}:
\begin{equation}
	c_n(\kappa) = \int_\SOn \exp\left( \kappa\,\trace(Z) \right) \dmu(Z),
\end{equation}
where $\dmu$ is the normalized Haar measure over $\SOn$ such that $\int_{\SOn}\mathrm{d}\mu(Z) = 1$. In particular, $c_n(0) = 1$. The integrand, $g(Z) = \exp\left( \kappa\,\trace(Z) \right)$, is a class function, meaning that for all $Q,Z\in\SOn$ we have $g(Z) = g(QZQ^{-1})$. We are thus in a position to use the Weyl integration formula specialized to $\SOn$~\citep[Exercise~18.1--2]{liegroupsbump}. Formula~\eqref{eq:weyl4} applies,
\begin{align}
	\int_{\mathrm{SO}(4)} g(Z)\, \dmu(Z) & = \frac{1}{4(2\pi)^2} \int_{-\pi}^{\pi}\! \int_{-\pi}^{\pi} \tilde g(\theta_1, \theta_2) \cdot |e^{i\theta_1}-e^{i\theta_2}|^2 \cdot |e^{i\theta_1}-e^{-i\theta_2}|^2 \, \mathrm{d} \theta_1 \mathrm{d} \theta_2,
	\label{eq:weyl4apdx}
\end{align}
where we defined
\begin{align}
	\tilde g(\theta_1, \theta_2) & \triangleq g\left(\diag\left(\begin{pmatrix}\cos\theta_1 & -\sin\theta_1 \\ \sin\theta_1 & \cos\theta_1 \end{pmatrix},\begin{pmatrix}\cos\theta_2 & -\sin\theta_2 \\ \sin\theta_2 & \cos\theta_2 \end{pmatrix}\right)\right).
\end{align}
This reduces the problem to a classical integral over the square---or really the torus---$[-\pi, \pi] \times [-\pi, \pi]$. Evaluating $\tilde g$ is straightforward:
\begin{align}
	\tilde g(\theta_1, \theta_2) = \exp\big(2\kappa \cdot [\cos\theta_1+\cos\theta_2]\big).
	\label{eq:gdiag}
\end{align}
Using trigonometric identities, we also get:
\begin{align}
	& |e^{i\theta_1}-e^{i\theta_2}|^2 \cdot |e^{i\theta_1}-e^{-i\theta_2}|^2 \nonumber \\
	& = 4\big(1-\cos(\theta_1-\theta_2)\big)\big(1-\cos(\theta_1+\theta_2)\big) \nonumber \\
		& = 4\big( 1 - \cos(\theta_1-\theta_2) - \cos(\theta_1+\theta_2) + \cos(\theta_1-\theta_2)\cos(\theta_1+\theta_2) \big) \nonumber \\
		& = 4\left( 1 - 2\cos\theta_1\cos\theta_2+\frac{1}{2}(\cos 2\theta_1 + \cos 2\theta_2) \right).
	\label{eq:prodabs}
\end{align}
Each cosine factor now only depends on one of the angles. Plugging~\eqref{eq:gdiag} and~\eqref{eq:prodabs} in~\eqref{eq:weyl4apdx} and using Fubini's theorem, we get:
\begin{align}
	c_4(\kappa) & = \frac{1}{2\pi}\int_{-\pi}^{\pi} e^{2\kappa\cos\theta_1} \cdot h(\theta_1) \ \mathrm{d}\theta_1,
	\label{eq:c4intbis}
\end{align}
with:
\begin{align}
	h(\theta_1) & = \frac{1}{2\pi}\int_{-\pi}^{\pi} e^{2\kappa\cos\theta_2} \left( 1+\frac{1}{2}\cos 2\theta_1 - 2\cos\theta_1 \cos\theta_2 + \frac{1}{2}\cos 2\theta_2 \right) \mathrm{d}\theta_2.
\end{align}
Now recalling the definition of the modified Bessel functions of the first kind~\citep{wolframbessel},
\begin{align}
	I_\nu(x) = \frac{1}{2\pi} \int_{-\pi}^{\pi} e^{x\cos\theta} \cos(\nu\theta) \, \mathrm{d}\theta,
	\label{eq:bessel}
\end{align}
we further simplify $h$ to get:
\begin{align}
	h(\theta_1) & = \left(1+\frac{1}{2}\cos 2\theta_1\right) \cdot I_0(2\kappa) - 2\cos\theta_1 \cdot I_1(2\kappa) + \frac{1}{2} \cdot I_2(2\kappa).
	\label{eq:h}
\end{align}
Plugging~\eqref{eq:h} in~\eqref{eq:c4intbis} and resorting to Bessel functions again, we finally obtain the practical formula~\eqref{eq:c4} for $c_4(\kappa)$:
\begin{align}
	c_4(\kappa) & = \left[I_0(2\kappa) + \frac{1}{2} I_2(2\kappa)\right] \cdot I_0(2\kappa) - 2I_1(2\kappa)\cdot I_1(2\kappa) + \frac{1}{2} I_0(2\kappa) \cdot I_2(2\kappa) \nonumber \\
               & = I_0(2\kappa)^2-2I_1(2\kappa)^2 + I_0(2\kappa)I_2(2\kappa).
\end{align}

For generic $n$, the necessary manipulations are very similar to the developments in this appendix. For $n = 2$ or $3$, the computations are easier. For $n=5$, the computations take up about the same space. For $n\geq 6$, the same technique will still work but gets quite cumbersome.

In~\citep[Appendix A.6]{chikuse2003statistics}, Chikuse describes how the normalization coefficients for Langevin distributions on $\On$ can be expressed in terms of hypergeometric functions with matrix arguments. One advantage of this method is that it generalizes to non-isotropic Langevin's. The method we demonstrated here, on the other hand, is tailored for our need (isotropic Langevin's on $\SOn$) and yields simple expressions in terms of Bessel functions---which are readily available in Matlab for example.

%% file: curvatureterms.tex
\section{Curvature terms}\label{apdx:curvatureterms}

\add{We compute the curvature terms from theorems {\ref{thm:CRBsubmanifold}} and {\ref{thm:CRBquotientmanifold}} for $n=2$ and $n=3$ explicitly.} We first treat $\p_A$~\eqref{eq:pA}, then $\p_\emptyset$~\eqref{eq:pempty}. We show that for rotations in the plane ($n=2$), the parameter spaces are flat, so that curvature terms vanish exactly. For rotations in space ($n=3$), we compute the curvature terms explicitly and show that they are on the order of $\mathcal{O}(\mathrm{SNR}^{-2})$, whereas dominant terms in the CRB are on the order of $\mathcal{O}(\mathrm{SNR}^{-1})$, for the notion of SNR proposed in Section~\ref{sec:crbsynch}. It is expected that curvature terms are negligible for $n\geq 4$ too for the same reasons, but we do not conduct the calculations.

\subsection{Curvature terms for $\p_A$}

The manifold $\p_A$~\eqref{eq:pA} is a (product) Lie group. Hence, the Riemannian curvature tensor $\mathcal{R}$ of $\p_A$ on the tangent space $\T_{\bfR}\p_A$ is given by a simple formula~\citep[Corollary~11.10, p.\,305]{oneill}:
\begin{align}
	\inner{\mathcal{R}(\bfX,\bfOmega)\bfOmega}{\bfX} & = \frac{1}{4} \|[\bfX, \bfOmega]\|^2,
	\label{eq:curvatureliegroups}
\end{align}
where $[\bfX, \bfOmega]$ is the Lie bracket of $\bfX = (X_1, \ldots, X_N)$ and $\bfOmega = (\Omega_1, \ldots, \Omega_N)$, two vectors (not necessarily orthonormal) in the tangent space $\T_{\bfR}\p_A$. Following~\citep[Theorem~4]{crbsubquot}, in order to compute the curvature terms for the CRB of synchronization on $\p_A$, we first need to compute
\begin{align}
	\mathbf{R_m}[\bfOmega', \bfOmega'] & \triangleq \expectt{\inner{\mathcal{R}(\bfX, P_A\bfOmega') P_A\bfOmega'}{\bfX}},
	\label{eq:RmwithPA}
\end{align}
where $\bfOmega'$ is any tangent vector in $\T_{\bfR}\p$ and $P_A\bfOmega'$ is its orthogonal projection on $\T_{\bfR}\p_A$. We expand $\bfX$ and $\bfOmega = P_A\bfOmega'$ using the orthonormal basis $(\bfxi_{k\ell})_{k=1\ldots N, \ell = 1\ldots d}$~\eqref{eq:basisTRP} of $\T_{\bfR}\p \supset \T_{\bfR}\p_A$:
\begin{align}
	\bfOmega & = \sum_{k,\ell}  \alpha_{k\ell} \bfxi_{k\ell} & \textrm{and} & & \bfX & = \sum_{k,\ell} \beta_{k\ell} \bfxi_{k\ell},
\end{align}
such that $\Omega_k = R_k\sum_\ell \alpha_{k\ell}E_\ell$ and $X_k = R_k\sum_\ell \beta_{k\ell}E_\ell$. Of course, $\alpha_{k\ell} = \beta_{k\ell} = 0\ \forall k\in A$. Then, since $$[\bfX,\bfOmega] = ([X_1,\Omega_1], \ldots, [X_N,\Omega_N]),$$ it follows that:
\begin{align}
	\mathbf{R_m}[\bfOmega,\bfOmega] & = \expectt{\frac{1}{4}\|[\bfX, \bfOmega]\|^2} = \expectt{\frac{1}{4}\sum_k \|[X_k, \Omega_k]\|^2} \\ & = \frac{1}{4}\sum_k \expectt{ \Big\|\sum_{\ell,s} \alpha_{k\ell}\beta_{ks} [E_\ell, E_s] \Big\|^2 }.
	\label{eq:Rmgeneric}
\end{align}
For $\bfX$ the tangent vector in $\T_{\bfR}\p_A$ corresponding to the (random) estimation error $\Log_{\bfR}(\hat \bfR)$, the coefficients $\beta_{k\ell}$ are random variables. The covariance matrix $C_A$~\eqref{eq:CA} is given in terms of these coefficients by:
\begin{align}
	(C_A)_{kk',\ell\ell'} \triangleq \expectt{\inner{\bfX}{\bfxi_{k\ell}}\inner{\bfX}{\bfxi_{k'\ell'}}} = \expectt{\beta_{k\ell}\beta_{k'\ell'}}.
\end{align}
The goal now is to express the entries of the matrix associated to $\mathbf{R_m}$ as linear combinations of the entries of $C_A$.

For $n=2$, of course, $\mathbf{R_m}\equiv 0$ since Lie brackets vanish thanks to the commutativity of rotations in the plane.

For $n=3$, the constant curvature of $\SOt$ leads to nice expressions, which we obtain now. Let us consider the orthonormal basis $(E_1, E_2, E_3)$ of $\mathfrak{so}(3)$~\eqref{eq:basisSOt}. Observe that it obeys $[E_1,E_2]=E_3/\sqrt{2}, [E_2,E_3]=E_1/\sqrt{2}, [E_3,E_1]=E_2/\sqrt{2}$. As a result, equation~\eqref{eq:Rmgeneric} simplifies and becomes:
\begin{multline}
	\mathbf{R_m}[\bfOmega,\bfOmega] = \frac{1}{8}\sum_k \expect\,\big\{ (\alpha_{k2}\beta_{k3}-\alpha_{k3}\beta_{k2})^2 + 
(\alpha_{k3}\beta_{k1}-\alpha_{k1}\beta_{k3})^2 + (\alpha_{k1}\beta_{k2}-\alpha_{k2}\beta_{k1})^2 \big\}.
	\label{eq:Rmfoo}
\end{multline}
We set out to compute the $dN\times dN$ matrix $R_m$ associated to the bi-linear operator $\mathbf{R_m}$ w.r.t.\ the basis~\eqref{eq:basisTRP}. By definition, $(R_m)_{kk',\ell\ell'} = \mathbf{R_m}[\bfxi_{k\ell},\bfxi_{k'\ell'}]$. Equation~\ref{eq:Rmfoo} readily yields the diagonal entries ($k=k', \ell=\ell'$). Using the polarization identity to determine off-diagonal entries,
\begin{align}
	(R_m)_{kk',\ell\ell'} = \frac{1}{4}\big( \mathbf{R_m}[\bfxi_{k\ell}+\bfxi_{k'\ell'},\bfxi_{k\ell}+\bfxi_{k'\ell'}] - \mathbf{R_m}[\bfxi_{k\ell}-\bfxi_{k'\ell'},\bfxi_{k\ell}-\bfxi_{k'\ell'}] \big),
\end{align}
it follows through simple calculations (taking into account the orthogonal projection onto $\T_{\bfR}\p_A$ that appears in~\eqref{eq:RmwithPA}) that:
\begin{align}
	(R_m)_{kk',\ell\ell'} & = \begin{cases} \frac{1}{8} \sum_{s\neq \ell} (C_A)_{kk,ss} & \textrm{ if } k=k'\notin A, \ell=\ell',\\-\frac{1}{8} (C_A)_{kk, \ell\ell'} & \textrm{ if } k=k'\notin A, \ell\neq \ell',\\ 0 & \textrm{ otherwise}.\end{cases}
	\label{eq:Rmanchors}
\end{align}
Hence, $R_m(C_A)$ is a block-diagonal matrix whose nonzero entries are linear functions of the entries of $C_A$. Theorem~\ref{thm:CRBsubmanifold} requires~\eqref{eq:Rmanchors} to compute the matrix $R_m(F_A^\dagger)$. Considering the special structure of the diagonal blocks of $F_A^\dagger$~\eqref{eq:FA} (they are proportional to $I_3$), we find that
\begin{align}
	R_m(F_A^\dagger) & = \frac{1}{4} \ddiag(F_A^\dagger) = \frac{3}{4} \ddiag(\mathcal{L}_A^\dagger)\otimes I_3,
\end{align}
where $\ddiag$ puts all off-diagonal entries of a matrix to zero. Thus, as the SNR goes up and hence as $\mathcal{L}_A^\dagger$ goes down, the curvature term $R_m(F_A^\dagger)F_A^\dagger + F_A^\dagger R_m(F_A^\dagger)$ in Theorem~\ref{thm:CRBsubmanifold} will become negligible compared to the main term in the CRB, $F_A^\dagger$.

\subsection{Curvature terms for $\p_\emptyset$}
\label{subsec:curvaturequotient}

The manifold $\p_\emptyset$~\eqref{eq:pempty} is a quotient manifold of $\p$. Hence, the Riemannian curvature tensor $\mathcal{R}$ of $\p_\emptyset$ is given by O'Neill's formula~\citep[Thm~7.47, p.\,213 and Lemma~3.39, p.\,77]{oneill}, showing that the quotient operation can only increase the curvature of the parameter space:
\begin{align}
	\inner{\mathcal{R}(\D\pi \bfX,\D\pi\bfOmega)\D\pi\bfOmega}{\D\pi \bfX} & = \frac{1}{4} \|[\bfX, \bfOmega]\|^2 + \frac{3}{4} \|[\bfX, \bfOmega]^\mathcal{V}\|^2,
	\label{eq:curvaturequotient}
\end{align}
where $\bfX, \bfOmega$ are horizontal vectors in $\HH_{\bfR} \subset \T_{\bfR}\p$ identified with tangent vectors to $\p_\emptyset$ via the differential of the Riemannian submersion $\D\pi(\bfR)$~\eqref{eq:submersion}, denoted simply as $\D\pi$ for convenience. The vector $[\bfX,\bfOmega]^\mathcal{V} \in \VV_{\bfR} \subset \T_{\bfR}\p$ is the vertical part of $[\bfX,\bfOmega]$, i.e., the component that is parallel to the fibers. Since in our case, moving along a fiber consists in changing all rotations along the same direction, $[\bfX,\bfOmega]^\mathcal{V}$ corresponds to the mean component of $[\bfX,\bfOmega]$:
\begin{align}
	[\bfX,\bfOmega]^\mathcal{V} = (R_1\omega, \ldots, R_N\omega), \textrm{ with } \omega = \frac{1}{N} \sum_{k=1}^N [R_k\transpose X_k^{}, R_k\transpose \Omega_k^{}].
\end{align}

For $n=2$, since $[\bfX,\bfOmega] = 0$, $[\bfX,\bfOmega]^\mathcal{V} = 0$ also, hence $\p_\emptyset$ is still a flat manifold, despite the quotient operation. We now show that for $n=3$ the curvature terms in Theorem~\ref{thm:CRBquotientmanifold} are equivalent to the curvature terms for $\p_A$ with $A:=\emptyset$ plus extra terms that decay as $1/N$ and can thus be neglected.

The curvature operator $\mathbf{R_m}$~\citep[eq.\,(54)]{crbsubquot} is given by:
\begin{align}
	\mathbf{R_m}[\bfxi_{k\ell}, \bfxi_{k\ell}] & \triangleq \expectt{\inner{\mathcal{R}(\D\pi \bfX, \D\pi \bfxi_{k\ell})\D\pi \bfxi_{k\ell}}{\D\pi \bfX}} \\
	& = \expectt{\frac{1}{4}\|[\bfX, \bfxi_{k\ell}^{} - \bfxi_{k\ell}^{\mathcal{V}}]\|^2 + \frac{3}{4}\|[\bfX, \bfxi_{k\ell}^{} - \bfxi_{k\ell}^{\mathcal{V}}]^{\mathcal{V}}\|^2}.
\end{align}
The tangent vector $\bfxi_{k\ell}^{} - \bfxi_{k\ell}^{\mathcal{V}}$ is, by construction, the horizontal part of $\bfxi_{k\ell}$. The vertical part decreases in size as $N$ grows: $\bfxi_{k\ell}^{\mathcal{V}} = \frac{1}{N}(R_1E_\ell, \ldots, R_NE_\ell)$.
It follows that:
\begin{align}
	\expectt{\|[\bfX, \bfxi_{k\ell}^{} - \bfxi_{k\ell}^{\mathcal{V}}]\|^2} = \expectt{\|[\bfX, \bfxi_{k\ell}]\|^2} (1 + \mathcal{O}(1/N)).
\end{align}
Hence, up to a factor that decays as $1/N$, the first term in the curvature operator $\mathbf{R_m}$ is the same as that of the previous section for $\p_A$, with $A:=\emptyset$. We now deal with the second term defining $\mathbf{R_m}$:
\begin{align}
	[\bfX, \bfxi_{k\ell}]^{\mathcal{V}} & = (R_1\omega, \ldots, R_N\omega), \textrm{ with} \\
	\omega & = \frac{1}{N} [R_k\transpose X_k^{}, E_\ell^{}] = \frac{1}{N} \sum_s \beta_{ks} [E_s, E_\ell^{}].
\end{align}
It is now clear that for large $N$ this second term is negligible compared to $\expectt{\|[\bfX, \bfxi_{k\ell}]\|^2}$:
\begin{align}
	\left\| [\bfX, \bfxi_{k\ell}]^{\mathcal{V}} \right\|^2 & = N\|\omega\|^2 = \mathcal{O}(1/N).
\end{align}
Applying polarization to $\mathbf{R_m}$ to compute off-diagonal terms then concludes the argument showing that the curvature terms in the CRB for synchronization of rotations on $\p_\emptyset$, despite an increased curvature owing to the quotient operation~\eqref{eq:curvaturequotient}, are very close (within a $\mathcal{O}(1/N)$ term) to the curvature terms established earlier for synchronization on $\p_A$, with $A:=0$. We do not include an exact derivation of these terms as it is quite lengthy and does not bring much insight to the problem.

%% file: GijGjiproof.tex
\section{Proof that $G_{ij}(QZQ\transpose) = QG_{ij}(Z)Q\transpose$ and that $G_{ij}(Z\transpose) = -G_{ij}(Z)$.}
\label{apdx:GijGji}

Recall the definition of $G_{ij} : \SOn \to \son$~\eqref{eq:defG} introduced in Section~\ref{sec:fimsynch}:
\begin{align}
	G_{ij}(Z) & = \left[\grad\,\log f_{ij}(Z)\right]\transpose Z.
\end{align}
We now establish a few properties of this mapping. Let us introduce a few functions:
\begin{align}
	g & : \SOn \to \mathbb{R} : Z \mapsto g(Z) = \log f_{ij}(Z), \\
	h_1 & : \SOn \to \SOn : Z \mapsto h_1(Z) = QZQ\transpose, \\
	h_2 & : \SOn \to \SOn : Z \mapsto h_2(Z) = Z\transpose.
\end{align}
Notice that because of Assumption~\ref{assu:invariant} ($f_{ij}$ is only a function of the eigenvalues of its argument), we have $g\circ h_i \equiv g$ for $i=1,2$. Hence,
\begin{align}
	\grad\,g(Z) = \grad(g\circ h_i)(Z) = (\D h_i(Z))^*\left[ \grad\,g(h_i(Z)) \right],
	\label{eq:gradggradgh}
\end{align}
where $(\D h_i(Z))^*$ denotes the adjoint of the differential $\D h_i(Z)$, defined by
\begin{align}
	\forall H_1, H_2 \in \T_Z\SOn, \quad \inner{\D h_i(Z)[H_1]}{H_2} = \inner{H_1}{(\D h_i(Z))^*[H_2]}.
	\label{eq:defadjoint}
\end{align}
The rightmost equality of~\eqref{eq:gradggradgh} follows from the chain rule. Indeed, starting with the definition of gradient, we have:
\begin{align}
	\forall H\in\T_Z\SOn, \quad \inner{\grad (g\circ h_i)(Z)}{H} & = \D (g\circ h_i)(Z)[H] \nonumber\\
											 & = \D g(h_i(Z))[\D h_i(Z)[H]] \nonumber\\
											 & = \inner{\grad\, g(h_i(Z))}{\D h_i(Z)[H]} \nonumber\\
											 & = \inner{(\D h_i(Z))^*\left[\grad\, g(h_i(Z)) \right]}{H}.
	\label{eq:gradgcirch}
\end{align}
Let us compute the differentials of the $h_i$'s and their adjoints:
\begin{align}
	\D h_1(Z)[H] & = QHQ\transpose, & (\D h_1(Z))^*[H] & = Q\transpose H Q, \\
	\D h_2(Z)[H] & = H\transpose, & (\D h_2(Z))^*[H] & = H\transpose.
\end{align}
Plugging this in~\eqref{eq:gradggradgh}, we find two identities (one for $h_1$ and one for $h_2$):
\begin{align}
	\grad\,\log f_{ij}(Z) & = Q\transpose [ \grad\,\log f_{ij}(QZQ\transpose) ] Q,\\
	\grad\,\log f_{ij}(Z) & = [ \grad\,\log f_{ij}(Z\transpose) ]\transpose.
\end{align}
The desired result about the $G_{ij}$'s now follows easily. For any $Q\in\On$,
\begin{align}
	G_{ij}(QZQ\transpose) & = [ \grad\,\log f_{ij}(QZQ\transpose) ]\transpose QZQ\transpose
								  = [ Q \grad\,\log f_{ij}(Z) Q\transpose ]\transpose QZQ\transpose 
								  = QG_{ij}(Z)Q\transpose;
\label{eq:Gijprop1}
\end{align}
and similarly:
\begin{align*}
	G_{ij}(Z\transpose) & = [ \grad\,\log f_{ij}(Z\transpose) ]\transpose Z\transpose 
								  = \grad\,\log f_{ij}(Z) Z\transpose 
								  = ZG_{ij}\transpose(Z)Z\transpose 
								  = -ZG_{ij}(Z)Z\transpose 
								  = -G_{ij}(Z),
\end{align*}
where we used that $G_{ij}(Z)$ is skew-symmetric and we used~\eqref{eq:Gijprop1} for the rightmost equality.

%% file: Pklproof.tex
\section{Proof of lemma~\ref{lem:Pkl}}
\label{apdx:Pkl}
\add{Lemma~{\ref{lem:Pkl}} essentially states that, given two orthogonal, same-norm vectors $E$ and $E'$ in $\son$, there exists a rotation which maps $E$ to $E'$. Applying that same rotation to $E'$ (loosely, rotating by an additional 90$^\circ$) recovers $-E$. This fact is obvious if we may use any rotation on the subspace $\son$. The set of rotations on $\son$ has dimension $d(d-1)/2$, with $d = \dim\son = n(n-1)/2$. In contrast, for the proof of Lemma~{\ref{lem:decorrelation}} to go through, we need to restrict ourselves to rotations of $\son$ which can be written as $\Omega \mapsto P\transpose \Omega P$, with $P\in\On$ orthogonal. We thus have only $d$ degrees of freedom. The purpose of the present lemma is to show that this can still be done if we further restrict the vectors $E$ and $E'$ as prescribed in Lemma~{\ref{lem:Pkl}}.}

\begin{proof}
We give a constructive proof, distinguishing among three cases. \emph{(Case 1: $\{i,j\}\cap\{k,\ell\}=\emptyset$)}. Construct $T$ as the identity $I_n$ with columns $i$ and $k$ swapped, as well as columns $j$ and $\ell$. Construct $S$ as $I_n$ with $S_{ii} := -1$. By construction, it holds that $T\transpose ET = E'$, $T\transpose E'T = E$, $SES = -E$ and $SE'S = E'$. Set $P = TS$ to conclude: $P\transpose EP = ST\transpose ETS = SE'S = E'$, $P\transpose E'P = ST\transpose E'TS = SES = -E$. \emph{(Case 2: $i=k,j\neq\ell$)}.  Construct $T$ as the identity $I_n$ with columns $j$ and $\ell$ swapped. Construct $S$ as $I_n$ with $S_{jj} := -1$. The same properties will hold. Set $P = TS$ to conclude. \emph{(Case 3: $i=\ell,j\neq k$)}.  Construct $T$ as the identity $I_n$ with columns $j$ and $k$ swapped and with $T_{ii} := -1$. Construct $S$ as $I_n$ with $S_{jj} := -1$. Set $P = TS$ to conclude. \emph{(Cases 4 and 5: $j=k$ or $j=\ell$)}. The same construction goes through.
\end{proof}